\documentclass[12pt,onecolumn]{IEEEtran}
\usepackage{makecell,amsbsy,amsmath,amssymb,epsfig,bbm,mathrsfs,multirow,amsthm,bm}
\usepackage{array,multirow,graphicx}
\usepackage[english]{babel}
\usepackage[justification=centering]{caption}
\usepackage{url}
\usepackage{threeparttable}
\usepackage{epstopdf}
\usepackage{subfigure}
\usepackage[ruled,linesnumbered]{algorithm2e}
\usepackage{url}

\DeclareMathAlphabet{\mathpzc}{OT1}{pzc}{m}{it}

\newtheorem{theorem}{\textbf{\textsc{Theorem}}}

\IEEEoverridecommandlockouts

\begin{document}
	\bibliographystyle{IEEE2}

\title{Dynamic Model for Network Selection in Next Generation HetNets with Memory-affecting Rational Users}

\author{Shaohan~Feng,~\IEEEmembership{Student Member,~IEEE,} 
	Dusit~Niyato,~\IEEEmembership{Fellow,~IEEE,} 
	Xiao~Lu,~\IEEEmembership{Student Member,~IEEE,} 
	Ping~Wang,~\IEEEmembership{Senior Member,~IEEE,} 
	and~Dong~In~Kim,~\IEEEmembership{Fellow,~IEEE}
\IEEEcompsocitemizethanks{\IEEEcompsocthanksitem S. Feng and D. Niyato are with the School of Computer Science and Engineering, Nanyang Technological University, Singapore. 
\IEEEcompsocthanksitem X. Lu is with the Department of Electrical and Computer Engineering, University of Alberta, Canada. 
\IEEEcompsocthanksitem P. Wang is with the Department of Electrical Engineering and Computer Science, Lassonde School of Engineering, York University, Canada. 
\IEEEcompsocthanksitem D. I. Kim is with the Department of Electrical and Computer Engineering, Sungkyunkwan University, Suwon 16419, Korea. D. I. Kim is the corresponding author of the paper.}}

\maketitle
\vspace{-15mm}
\begin{abstract}\vspace{-2mm}
Recently, due to the staggering growth of wireless data traffic, heterogeneous networks have drawn tremendous attention due to the capabilities of enhancing the capacity/coverage and to save energy consumption for the next generation wireless networks. In this paper, we study a long-run user-centric network selection problem in the 5G heterogeneous network, where the network selection strategies of the users can be investigated dynamically. Unlike the conventional studies on the long-run model, we incorporate the memory effect and consider the fact that the decision-making of the users is affected by their memory, i.e., their past service experience. Namely, the users select the network based on not only their instantaneous achievable service experience but also their past service experience within their memory. Specifically, we model and study the interaction among the users in the framework of fractional evolutionary game based on the classical evolutionary game theory and the concept of the power-law memory. We analytically prove that the equilibrium of the fractional evolutionary game exists, is unique and uniformly stable. We also numerically demonstrate the stability of the fractional evolutionary equilibrium. Extensive numerical results have been conducted to evaluate the performance of the fractional evolutionary game. The numerical results have revealed some insightful findings. For example, the user in the fractional evolutionary game with positive memory effect can achieve a higher cumulative utility compared with the user in the fractional evolutionary game with negative memory effect. Moreover, the fractional evolutionary game with positive memory effect can reduce the loss in the user's cumulative utility caused by the small-scale fading.
\end{abstract}

\begin{IEEEkeywords}
Network selection, fractional evolutionary game, memory-affecting rationality, and heterogeneous network.
\end{IEEEkeywords}

\section{Introduction}
\label{sec:introduction}

Due to the proliferation of wireless handsets and portable devices as well as data-hungry multimedia applications, mobile data demand continues to grow exponentially in recent years and is likely to soon outgrow the capacity of the current cellular networks~\cite{Cambridge2014Joint}. To address this severe issue, the network will continue to become increasingly heterogeneous as we move to fifth generation (5G)~\cite{andrews2014will}. A heterogeneous network (HetNet) is a wireless network consisting of nodes with different transmission powers and coverage sizes~\cite{hu2014energy}. High power nodes (HPNs) with large coverage areas are deployed in a planned way for blanket coverage of urban, suburban, or rural areas. Low power nodes (LPNs) with small coverage areas aim to complement the HPNs for coverage extension and throughput enhancement. By taking advantage of the best of different networking technologies, multi-faceted benefits can be reaped in 5G heterogeneous networks, such as improving the utilization of the network resources, enhancing the scalability and provisioning networking service upon requirement~\cite{hu2015mih}.

%\clearpage

\subsection{Motivation}

In this paper, we study a user-centric network selection problem in 5G HetNets, where the users can freely select the network and access the networking service of which. Moreover, we study the network selection problem on a long-run basis such that the behaviors of the users can be investigated dynamically. It is worth noting here that to study such a long-run network selection problem, it is natural and practical to consider that the users' memory, i.e., past service experience, will affect their decision-making. In other words, the users make their decisions by taking into account not only their instantaneous achievable service experience but also their past service experience. In brief, the action that the user selects the network can be regarded as the behavior of an economic agent in an economic process. In the economic process, the economic agent is aware of and pays great attention to the history of this process, hence the impact of which on the behavior cannot be ignored. For example, one recent study reported that the visitors are probably not deciding whether to enrol in the visited school during their visit, cloudiness must be influencing college decisions through memory~\cite{simonsohn2009weather}. The reason is that, in reality, the students' impression, i.e., their past experience in memory, about the college affects their decisions on the college enrollment.

%\clearpage

\subsection{Our Contributions}
\label{subsec:contribution}

Nevertheless, the conventional dynamical model, i.e., classical evolutionary game, cannot capture the impact of the users' memory on their decision-making due to the fact that the players in the classical evolutionary game only consider the instantaneous achievable utility~\footnote{In economics, the agents are memory-aware due to the well-known fact that the agents can remember the history of the economic processes~\cite{sun2018new}.}. In this case, we incorporate the concept of the power-law memory~\cite{tarasova2018concept}, which is used to depict the impact of the users' memory on their strategies~\cite{lynch1991memory}. That is, whenever the users are making decisions on their current strategies, they will take into consideration not only their instantaneous achievable utility but also their previous decisions within the memory~\cite{tarasova2017logistic}. It is worth noting that extensive works that incorporate the memory effect have been proposed in the economics such as~\cite{tarasova2016fractional, tarasov2016long} and the modern physics such as~\cite{koeller1984applications, tarasov2011fractional}. In particular, a generalization of the economic model of natural growth, which takes into account the power-law memory effect, is suggested by the authors in~\cite{tarasova2016fractional}. Regarding the applications in the modern physics, the memory effect has been incorporated in~\cite{koeller1984applications} to study the materials with memory.

%\clearpage 

In this paper, we study a dynamic network selection problem in 5G HetNets as shown in Fig.~\ref{fig:system_model}, where a HetNet constituted of ultra-high frequency (UHF), i.e., the frequencies below $6$ GHz, base station (BS), millimeter-wave (mmWave) BS, and unmanned aerial vehicle (UAV)-enabled mmWave BS is considered to be the application scenario for the 5G HetNet. In the problem, there two parties, i.e., different types of BSs working as the utility providers with a flat-rate pricing scheme for provisioning the communication service and the communication service customers, i.e., memory-affecting rational users. Specifically, we first formulate a classical evolutionary game to analyze the interaction among the users in the 5G HetNet on a long-run basis. Then, by incorporating the concept of the power-law memory, which is depicted by using the fractional calculus (including fractional derivatives and integrals), we cast the classical evolutionary game as a fractional evolutionary game, where the dynamic behaviors of the memory-affecting rational users can be investigated. 
%\clearpage

The major contributions of this paper are summarized as follows:
\begin{itemize}
	\item We model the interactions among the network users in a framework of the classical evolutionary game. Different from the traditional static game models in the existing literature, the proposed framework in this paper is characterized by modeling the dynamic, i.e., time-variant, behaviors of the users on a long-term basis.
	
	\item For the proposed classical evolutionary game model, we further incorporate the concept of power-law memory to reformulate it into a fractional evolutionary game. The network selection strategies of the users from the classical evolutionary game and that from the fractional evolutionary game are compared to investigate the impact of the users' memory on their strategies.
	
	\item We theoretically prove that the equilibrium in the fractional evolutionary game exists, and is unique and uniformly stable. Moreover, we numerically verify the stability of the equilibrium by using the direction field of the replicator dynamics. 
	
	\item Extensive simulations have been conducted to evaluate the performance of the proposed fractional evolutionary game. The numerical results have revealed some interesting findings. For example the users in the fractional evolutionary game with positive memory effect can achieve higher utility compared with that in the classical evolutionary game and the fractional evolutionary game with negative memory effect.
\end{itemize}

%\clearpage

The rest of the paper is organized as follows. Section~\ref{sec:related_work} presents the related work and highlights the research gap in the literature. Section~\ref{sec:preliminary} introduces the network model and the concept of the power-law memory. Section~\ref{sec:system_model} describes the system model and the classical evolutionary game formulation as well as the fractional evolutionary game formulation. Section~\ref{sec:equilibrium_analysis} provides the proofs of the existence and uniqueness as well as the stability of the equilibrium. Section~\ref{sec:performance} presents the numerical performance evaluation with some insightful results. Finally, Section~\ref{sec:conclusion} concludes the paper.

%\clearpage

\section{Related Work}
\label{sec:related_work}

\subsection{Heterogeneous Networks}

Nowadays, due to the capability of improving the spectral efficiency, the HetNets are inevitably becoming an alternative solution in helping to meet the exponentially increased wireless data traffic~\cite{cai2016green}. The HetNet, which is an overlay of multiple cellular networks such as macrocells, microcells, and femtocells, has been verified to have greater end-user data rate and throughput as well as better indoor and cell-edge coverage~\cite{madhusudhanan2012downlink}. For these reasons, a number of works are presented to improve the performance of the HetNets. In~\cite{elshaer2016downlink}, the authors developed a general analytical model for a hybrid cellular network constituted of the traditional sub-6 GHz macrocells and mmWave small cells and analyzed how the user should associate with these two types of BSs in the uplink and downlink. To reduce the number of handoffs while maintaining user's Quality of Service (QoS) requirements in mmWave HetNets, a reinforcement learning based handoff policy named SMART was introduced in~\cite{sun2018smart}. The authors in~\cite{tan2018learning} proposed a joint resource allocation and network access problem to investigate the coexistence mechanism for license-assisted access LTE (LAA-LTE) based HetNets. In this joint resource allocation and network access problem, the normalized throughput of the unlicensed band was maximized while meeting the QoS requirement of incumbent WiFi user. Considering the multimedia application QoS in the heterogeneous wireless networks, an optimal distributed network selection scheme was presented in~\cite{si2010optimal}, which is applicable to both tight coupling and loose coupling scenarios in the integration of heterogeneous wireless network. The author in~\cite{xie2016joint} studied a joint user association and rate allocation problem for HTTP adaptive streaming in HetNets, where the system utility was maximized, and moreover, the user's requirement of Quality of Experience (QoE) was satisfied. By deploying distributed caching helpers in heterogeneous architectures, the capacity bottleneck of backhaul networks can be alleviated while the area spectral efficiency can be improved~\cite{li2017cooperative}.

%\clearpage

\subsection{The Network Selection in Wireless Networks}

Recently, due to the diversity of the communication technologies and also the heterogeneity of the cellular networks, there are some works concerning the network selection in the wireless networks. For example, the authors in~\cite{wang2013mathematical} presented a tutorial of mathematical modeling for network selection in HetNets. In~\cite{jiang2012renewal}, a renewal theoretical framework was proposed to study the dynamic spectrum access in cognitive radio networks. Therein, the authors considered a challenging scenario that the secondary user does not know about the primary user's communication mechanism, which would induce additional interference in the primary user's communication. A stochastic game was proposed in~\cite{yang2013wireless} to investigate the network selection in the wireless access networks. In this stochastic game, the influence of subsequent users' network selection decision on an individual's throughput due to the limited available resources is captured by using the negative network effect. The authors in~\cite{chen2013evolutionarily} designed spectrum access mechanisms with the scenarios of both complete and incomplete network information, where an evolutionary spectrum access mechanism and a distributed learning mechanism were introduced to handle the scenarios with complete and incomplete network information, respectively. A new MAC algorithm was proposed in~\cite{rad2009utility} to achieve the network utility maximization while taking the form of random access without message passing. Hop selection problem in ultra-dense cellular networks was discussed by the authors in~\cite{ge20165g}. 

%\clearpage

\subsection{Memory effect}

As one of the important rational components, memory formalism in finance represents the effect of memory on the economic operators for their action in the markets~\cite{forte2014handbook}. Specifically, the economic process with dynamic memory assumes the awareness of the economic agents about the history of this process. In such a process, the behavior of the economic agent is based on not only the instantaneous information of the state of the process at a given moment but also the information about the process states at previous moments over a time interval~\cite{tarasova2018concept}. This is well consistent with the reality, which results in that the economic processes with dynamic memory are actively studied in recent years. In~\cite{tarasova2016elasticity}, generalizations of price elasticity of demand to the case of processes with dynamic memory have been defined, where the changes of the price in the previous time horizon are taken into account. The authors in~\cite{tarasova2017economic} discussed a generalization of the economic growth model with a constant pace under the effect of dynamic memory. In this model, the economic agents are considered to have the memory of their previous decisions and states, which will consequently induce different reactions for the agents.

However, to our best knowledge, it still remains an open research question that what will happen if the decision-making of the users is affected by their memory during the network selection processes in the HetNets. The aforementioned works inspire us to incorporate the users' memory and investigate their interaction through a fractional evolutionary game theory framework. This is the main objective as well as the major contribution of our studies.

%\clearpage

\section{Preliminary}
\label{sec:preliminary}

\begin{figure}[!]
	\centering
	\includegraphics[width=0.6\textwidth,trim=50 130 700 430,clip]{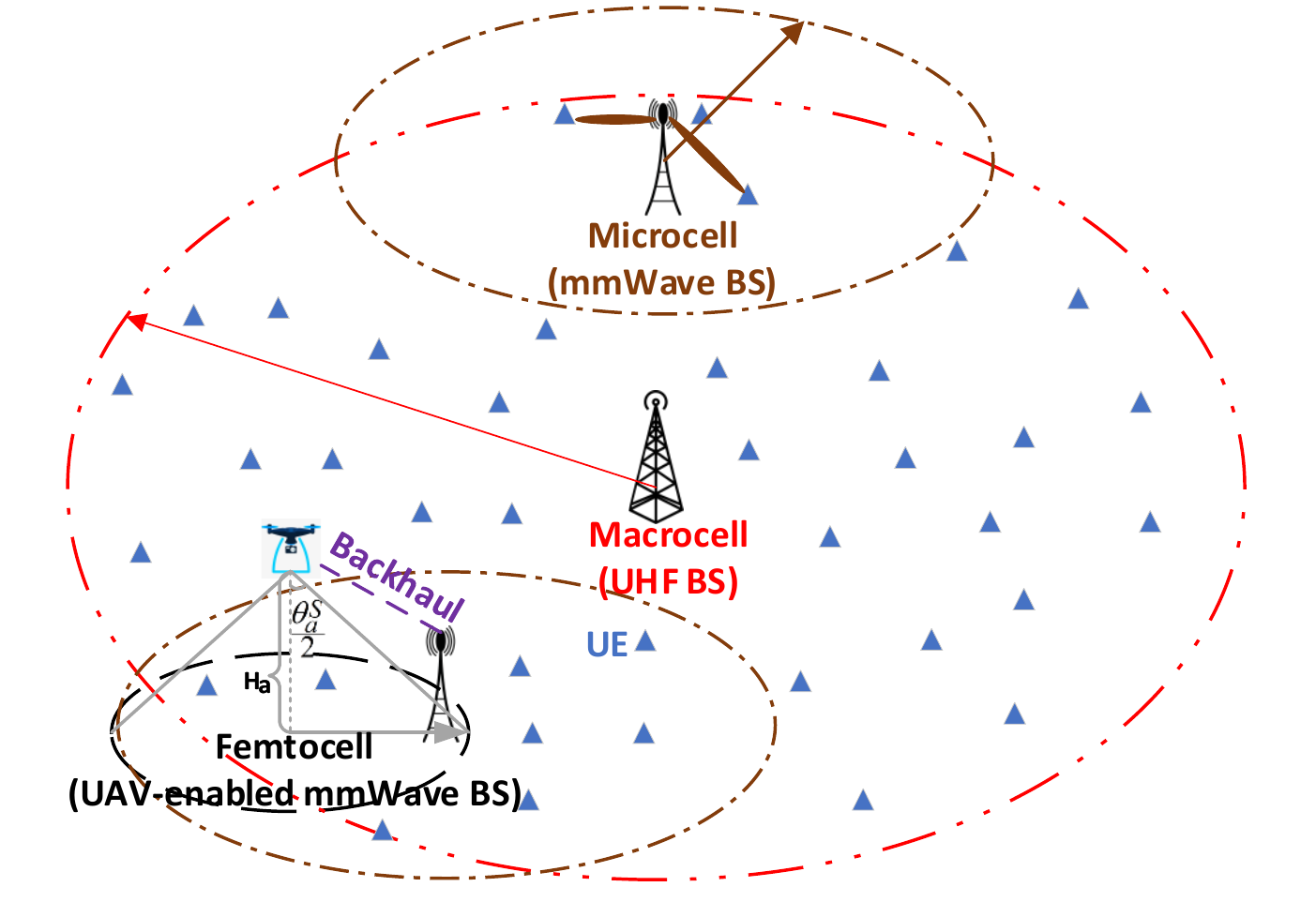}
	\caption{5G Heterogeneous Network Architecture}
	\label{fig:system_model}
\end{figure}

As shown in Fig.~\ref{fig:system_model}, we consider a 5G HetNet constituted of three types of BSs, i.e., UHF BS, mmWave BS, and UAV-enabled mmWave BS as the application scenario. Here, the propagation models of the aforementioned types of BSs are presented in Section~\ref{subsec:network_model}. For simplicity, we assume that ideal backhaul links exist between the UAV-enabled mmWave BSs and their nearby mmWave BSs. As the mmWave BSs' transmit power is much larger than that of the UAV-enabled mmWave BSs, it is reasonable to consider that the transmission rate of the backhaul links of the UAV-enabled mmWave BSs is much faster than that of the downlinks of the UAV-enabled mmWave BSs. Therefore, the throughput of the backhauls of the UAV-enabled mmWave BSs will not affect the throughput of their downlinks. In addition, as most of the energy of the UAV-enabled mmWave BS will be consumed by its flight~\cite{zeng2019energy}, we assume that the scheduling of the UAV-enabled mmWave communication service can ensure that the UAV-enabled mmWave BS will be replaced by a fully charged UAV-enabled mmWave BS before its battery is going to be depleted by the flight. Therefore, the limited energy of the UAV-enabled mmWave BS will not affect the UAV-enabled mmWave communication service provision. Note here that the service customer, i.e., the user, is considered to have an omnidirectional antenna.

Additionally, to incorporate the concept of power-law memory, we introduce the left-sided Caputo fractional derivative with respect to time~\cite{tarasova2017logistic} in Section~\ref{subsec:caputo_derivative}. The power-law memory can be captured by the left-sided Caputo fractional derivative, which has been widely adopted in the literature such as~\cite{tarasova2017logistic, tarasova2016fractional, tarasova2016elasticity}. 

%\clearpage

\subsection{Network Model}
\label{subsec:network_model}

\begin{table*}[!]\label{tab:notation_network}
	\centering
	\caption{Notations for the Network Model}
	\begin{tabular}{|c|l|}
		\hline
		\hline
		{\bf{Symbol}} & {\bf{Definition}} \\
		\hline
		$W_m$, $W_u$ $W_a$ & Bandwidth of an mmWave channel, a UHF channel, and an UAV-enabled mmWave channel, respectively.\\
		\hline
		$P^{\rm{t}}_u$, $P^{\rm{t}}_m$, $P^{\rm{t}}_a$ & The transmit power of a UHF BS, an mmWave BS, and an UAV-enabled mmWave BS, respectively. \\
		\hline
		$\sigma^2_m$, $\sigma^2_u$, $\sigma^2_a$ & The noise power in the mmWave channel, UHF channel, and UAV-enabled mmWave channel, respectively.\\
		\hline
		$\alpha_u$, $\alpha_m^{\text{LOS}}$, $\alpha_m^{\text{NLOS}}$ & Path-loss exponent of UHF signals, LOS and NLOS path-loss exponent of mmWave signals. \\
		\hline
		$f_m$, $f_u$ & mmWave and UHF carrier frequencies. \\
		\hline
		$G_u$, $G_{m,i}$, $G_{a,i}$ & \makecell[l]{The gains of UHF BS antenna, mmWave BS antenna, and UAV-enabled mmWave BS antenna, respectively.} \\
		\hline
		$G^M$, $G^S$, $G^A$& \makecell[l]{mmWave BS's main lobe gain, mmWave BS's side lobe gain, and UAV-enabled mmWave BS's main lobe\\ gain, respectively.}\\
		\hline
		$C$, $D$ & Fractional LOS area $C$ in a disc of radius $D$. \\
		\hline
		$h_{u,i}$, $h_{m,i}$, $h_{a,i}$ & \makecell[l]{small-scale fading of UHF BS antenna, mmWave BS antenna, and UAV-enabled mmWave BS\\ antenna, respectively.} 	\\
		\hline
		$N_u$, $N_m$, $N_a$ & Number of users served by a UHF BS, an mmWave BS, and an UAV-enabled mmWave BS, respectively.	\\
		\hline
		$H_a$ & The altitude of UAV-enabled mmWave BS $a$. \\
		\hline
		$\theta_m^S$, $\theta^S_a$ & \makecell[l]{The main beamwidth of mmWave BS $m$ and the half-power beamwidth of UAV-enabled mmWave\\ BS $a$, respectively.} 	\\
		\hline
		${\cal{C}}^{\rm{U}}_i$, ${\cal{C}}^{\rm{M}}_i$, ${\cal{C}}^{\rm{A}}_i$ & The sets of UHF BSs, mmWave BSs and UAV-enabled mmWave BSs cover user $i$, respectively.	\\
		\hline
	\end{tabular}
\end{table*}

\subsubsection{UHF Propagation Model}
\label{subsubsec:uhf_propagation}

A set of UHF BSs, denoted by ${\cal{U}}$, is deployed to provide UHF communication service. The received downlink signal power at user $i\in{\cal{N}}$ from UHF BS $u\in{\cal{U}}$ is 
\begin{equation}
%\small{
P_{u,i}=P^{\rm{t}}_uh_{u,i}\rho_uG_u\left[L_u\left(l_u-l_i\right)\right]^{-1},
%}
\end{equation}
where $l_u$, $l_i\in{\mathbb{R}}^{3\times1}$ respectively denote the locations of UHF BS $u$ and user $i$\footnote{The third components of the three-dimensional vectors $l_u$ and $l_i$ indicate the altitudes of UHF BS $u$ and user $i$, respectively.}, $L_u\left(z\right)=\left\|z\right\|^{\alpha_u}$ is the path-loss function, $h_{u,i}$ is the small-scale fading, $G_u$ is the antenna gain, $\rho_u$ is the near-field path loss at $1$ m, i.e., $\rho_u=\left(\frac{c}{4\pi f_u}\right)^2$, $c$ represents the speed of light, and $P^{\rm{t}}_u$ is the transmit power. Let ${\cal{I}}_u$ denote the set of the UHF BSs using the same channel as that of UHF BS $u$, the co-channel interference of UHF BS $u$ is ${\sum\limits_{w\in{\cal{I}}_u} P^{\rm{t}}_wh_{w,i}\rho_wG_w\left[L_w\left(l_w-l_i\right)\right]^{-1}}$, and the SINR for user $i$ at UHF BS $u$ is $\frac{P^{\rm{t}}_uh_{u,i}\rho_uG_u\left[L_u\left(l_u-l_i\right)\right]^{-1}} {\sum\limits_{w\in{\cal{I}}_u} P^{\rm{t}}_wh_{w,i}\rho_wG_w\left[L_w\left(l_w-l_i\right)\right]^{-1}+\sigma^2_u}$. As a result, the downlink transmission rate of user $i$ at UHF BS $u$ is given as
\begin{equation}\label{eq:UHF_rate}
%\small{
\begin{aligned}
&R_{u,i}=\frac{W_u}{N_u}\log_2\left(1+ \frac{P_{u,i}} {\sum\limits_{w\in{\cal{I}}_u} P_{w,i}+\sigma^2_u}\right)\\
=&\frac{W_u}{N_u}\log_2\left(1+ \frac{P^{\rm{t}}_uh_{u,i}\rho_uG_u\left[L_u\left(l_u-l_i\right)\right]^{-1}} {\sum\limits_{w\in{\cal{I}}_u} P^{\rm{t}}_wh_{w,i}\rho_wG_w\left[L_w\left(l_w-l_i\right)\right]^{-1}+\sigma^2_u}\right),
\end{aligned}
%}
\end{equation}
where $W_u$ is the bandwidth of UHF channel at UHF BS $u$ and $N_u$ is the number of users selecting UHF BS $u$. Here, we consider that UHF BSs use TDMA, and hence the number of users selecting UHF BS $w\in{\cal{I}}_u$ will not affect the co-channel interference received by the users in the coverage of UHF BS $u$. Note here that $R_{u,i}$ represents the per-user transmission rate averaged over frame time.

\subsubsection{mmWave Propagation Model}
\label{subsubsec:mmwave_propagation}

The set of mmWave BSs, denoted by ${\cal{M}}$, provide mmWave communication service. The received signal power at user $i$ in the downlink from an mmWave BS $m\in{\cal{M}}$ is 
\begin{equation}\label{eq:mmwave_propagation}
%\small{
P_{m,i}=p^{\rm{LOS}}_{m,i}P^{\rm{t}}_mh_{m,i}\rho_mG_{m,i}\left[L_m\left(l_m-l_i\right)\right]^{-1},
%}
\end{equation}
where $l_m \in {\mathbb{R}}^{3\times1}$ is the location of mmWave BS $m$\footnote{The third component of $l_m$ indicates the altitude of mmWave BS $m$.}, $L_m\left(z\right)=\left\|z\right\|^{\alpha_m^s}$ is the path-loss function for an mmWave channel, therein $s\in\left\{{\text{LOS}}, {\text{NLOS}}\right\}$ is the link indicator, $h_{m,i}$ is the small-scale fading, $G_{m,i}$ is the antenna gain, $\rho_m$ is the near-field path loss at $1$m, i.e., $\rho_m=\left(\frac{c}{4\pi f_m}\right)^2$, and $P^{\rm{t}}_m$ is the transmit power. $p^{\rm{LOS}}_{m,i}$ is the line-of-sight probability as a function of the distance between user $i$ and mmWave BS $m$, i.e., $r_{m,i}=\left\|l_m-l_i\right\|$, given by~\cite{singh2015tractable}
\begin{equation}
%\small{
p^{\rm{LOS}}_{m,i}=\left\{
\begin{aligned}
C,\quad &{\text{if}}\; r_{m,i} \le D,&\\
0,\quad &{\text{otherwise}},&
\end{aligned}
\right.
%}
\end{equation}
where $C\in\left[0,1\right]$ can be interpreted as the average LOS area in the spherical region around a typical user, e.g., $\left[C,D\right]=\left[0.081,250\right]$ for Chicago and $\left[C,D\right]=\left[0.117,200\right]$ for Manhattan~\cite{wang2016physical}. The mmWave BSs are equipped with directional antennas and the antenna gain of mmWave BS $m$ for user $i$, i.e., $G_{m,i}$, is given by 
\begin{equation}
%\small{
G_{m,i}\left(\theta_{m,i}\right)=\left\{
\begin{aligned}
G^{\rm{M}},\quad &{\text{if}}\;\left|\theta_{m,i}\right|\le\frac{\theta^{\rm{S}}_m}{2},&\\
G^{\rm{S}},\quad &{\text{otherwise}},&
\end{aligned}
\right.
%}
\end{equation}
where $\theta_{m,i}$ is the user $i$'s angle with respect to the best beam alignment, and $\theta_m^{\rm{S}}$ represents the main beamwidth of mmWave BS $m$. With the directional antennas, the SNR for user $i$ at mmWave BS $m$ is  $\frac{p^{\rm{LOS}}_{m,i}P^{\rm{t}}_mh_{m,i}\rho_mG_{m,i}\left[L_m\left(l_m-l_i\right)\right]^{-1}} {\sigma_m^2}$, where $\sigma^2_m$ is the noise power in the mmWave BS $m$'s channel. Correspondingly, the downlink transmission rate of user $i$ at mmWave BS $m$ is 
\begin{equation}\label{eq:mmwave_rate}
%\small{
\begin{aligned}
&R_{m,i}=\frac{W_m}{N_m}\log_2\left(1+\frac{P_{m,i}} {\sigma_m^2}\right)\\
=&\frac{W_m}{N_m}\log_2\left(1+\frac{p^{\rm{LOS}}_{m,i}P^{\rm{t}}_mh_{m,i}\rho_mG_{m,i}\left[L_m\left(l_m-l_i\right)\right]^{-1}} {\sigma_m^2}\right),
\end{aligned}
%}
\end{equation}
where $W_m$ is the bandwidth of an mmWave channel at mmWave BS $m$ and $N_m$ is the number of users selecting mmWave BS $m$.

\subsubsection{UAV-enabled mmWave Propagation Model} 
\label{subsubsec:uav_propagation}

The set of UAV-enabled mmWave BSs, denoted by ${\cal{A}}$, is deployed to provide UAV-enabled mmWave communication service. The received signal power at user $i$ in the downlink from an UAV-enabled mmWave BS $a\in{\cal{A}}$ is 
\begin{equation}
%\small{
P_{a,i}=p^{\rm{LOS}}_{a,i}P^{\rm{t}}_ah_{a,i}\rho_aG_{a,i}\left[L_a\left(l_a-l_i\right)\right]^{-1},
%}
\end{equation}
where the definitions of $P^{\rm{t}}_a$, $h_{a,i}$, $\rho_a$, and $L_a\left(z\right)$ are similar to that in~(\ref{eq:mmwave_propagation}). Here, $l_a\in{\mathbb{R}}^{3\times1}$ is the location of UAV-enabled mmWave BS $a$, and the altitude of UAV-enabled mmWave BS $a$ is denoted by $H_a$, which is the third component of $l_a$. Similar to~\cite{lyu2017blocking}, each UAV-enabled mmWave BS $a\in{\cal{A}}$ is equipped with a directional antenna pointing downward at the ground, whose half-power beamwidths are $\theta_{a}^{S}$ radians with $\frac{\theta_{a}^{S}}{2}\in\left(0, \frac{\pi}{2}\right)$. Then, the corresponding antenna gain for user $i$ at UAV-enabled mmWave BS $a$ is~\cite{balanis2016antenna}
\begin{equation}
%\small{
G_{a,i}\left(l_a,l_i\right)=\left\{
\begin{aligned}
G^A, \quad& {\text{if}} \; \left\|\left.l_a\right|_{H_a=0}-l_i\right\|\le H_a \tan \frac{\theta_{a}^{S}}{2},&\\
0, \quad &{\text{otherwise}}.&
\end{aligned}
\right.
%}
\end{equation}
$p^{\rm{LOS}}_{a,i}$ is the line-of-sight probability as a function of the distance between the user $i$ and UAV-enabled mmWave BS $a$ on the 2-D plane, i.e., $r_{a,i}=\left\|\left.l_a\right|_{H_a=0}-l_i\right\|$, as well as the altitude of UAV-enabled mmWave BS $a$, i.e., $H_a$, given by
$
p^{\rm{LOS}}_{a,i}\left(r_{a,i},H_a\right)=\frac{1}{1+b\exp\left(-c\left(\frac{180}{\pi}\tan^{-1}\left(\frac{H_a}{r_{a,i}}\right)-b\right)\right)}
$,
where $b$ and $c$ are the constants depending on the environment, and $\left.l_a\right|_{H_a=0}$ is the projection of UAV-enabled mmWave BS $a$ on the ground~\cite{al2014optimal}~\footnote{Note that due to the different propagation environments, the LOS models of the mmWave BS and the UAV-enabled mmWave BS are different.}. Due to the directional antenna equipped by the UAV-enabled mmWave BSs, the corresponding expected SNR of user $i$ at UAV-enabled mmWave BS $a$ is $\frac{ p^{\rm{LOS}}_{a,i} P^{\rm{t}}_ah_{a,i}\rho_aG_{a,i}\left[L_a\left(l_a-l_i\right)\right]^{-1}} {\sigma_a^2}$, where $\sigma_a^2$ is the noise power in UAV-enabled mmWave BS $a$'s channel. As a result, the downlink transmission rate of user $i$ at UAV-enabled mmWave BS $a$ is
\begin{equation}\label{eq:uav_rate}
%\small{
\begin{aligned}
&R_{a,i}=\frac{W_a}{N_a}\log_2\left(1+ \frac{P_{a,i}} {\sigma_a^2}\right)\\
=&\frac{W_a}{N_a}\log_2\left(1+ \frac{p^{\rm{LOS}}_{a,i}P^{\rm{t}}_ah_{a,i}\rho_aG_{a,i}\left[L_a\left(l_a-l_i\right)\right]^{-1}} {\sigma_a^2}\right),
\end{aligned}
%}
\end{equation}
where $W_a$ is the bandwidth of mmWave channel at UAV-enabled mmWave BS $a$ and $N_a$ is the number of users selecting UAV-enabled mmWave BS $a$.

%\clearpage

\subsection{The Concept of the Power-law Memory}
\label{subsec:caputo_derivative}

As introduced in~\cite{tarasova2018concept}, most of the economic process are memory-aware due to the fact the memory plays an essential role not only in the psychology but also in modern physics. To study the memory-aware economic processes, fractional calculus, which are the integrals and derivatives of non-integer order, have been widely introduced as a promising approach to investigate the role of the memory in the economic processes such as~\cite{wang1965principle} and~\cite{coleman1968general}. Specifically, the memory-aware economic processes are derived as follows:
\begin{enumerate}
	\item Basically, the most general formulation of the economic processes can be presented in a symbolic expression as $Y \left(t\right) = F_0^t \left(X \left(\tau\right)\right) + Y_0$, where $X\left(\tau\right)$ with $\tau\in\left[0,t\right]$ is the time-dependent input, $Y\left(t\right)$ is the time-dependent output, and $Y_0$ is the initial state of the time-dependent output of the processes. Therein, $F_0^t$ is an operator describing the relationship between the time-dependent input $X\left(\tau\right)$ and output $Y\left(t\right)$ and hence is a mapping that can determine the time-dependent output $Y\left(t\right)$ based on the time-dependent input $X\left(\tau\right)$ with $\tau \in \left[0,t\right]$. Initially, an integer-order integral based approach, i.e., $Y \left(t\right) = \int^t_0 X \left(\tau\right) {\rm{d}}t + Y_0$ with $F^t_0 \left(X \left(\tau\right)\right) := \int_0^t X \left(\tau\right) {\rm{d}}t$ , has been widely adopted to study the dynamics in such economic processes, e.g.,~\cite{niyato2008dynamics} and~\cite{semasinghe2014evolutionary}. 
	
	\item However, a vital drawback exists in the integer-order integral based approach. By taking derivative of $Y\left(t\right) = \int\limits_0^tX\left(\tau\right){\rm{d}}t + Y_0$ with respect to $t$, the formulation of the economic processes based on the integer-order integral can be represented in the expression of the ordinary differential equation as $\frac{{\rm{d}}}{{\rm{d}}t}Y\left(t\right) = X\left(t\right)$ with $Y\left(0\right) = Y_0$. In these processes, $\frac{{\rm{d}}}{{\rm{d}}t}Y\left(t\right)$ does not have any information about $X\left(\tau\right)$ for all $\tau\in\left[0,t\right)$ and hence depends only on $X\left(t\right)$, which means that the reaction of the economic processes based on the integer-order integral is not aware of the changes of $X\left(\tau\right)$ for all $\tau\in\left[0,t\right)$ and further not memory-aware.
	
	\item To address this drawback, the economists in~\cite{tarasova2018concept} considered the Volterra operator, i.e., a mapping, which can be specifically expressed as $F_0^t\left(X\left(\tau\right)\right) := \int\limits_0^tM_\beta\left(t - \tau\right)X\left(\tau\right){\rm{d}}t$, where $M_\beta\left(t - \tau\right)$ is a weighting function to measure the impact level of the previous input $X\left(\tau\right)$ on the current output $Y\left(t\right)$ according to the time distance between $\tau$ and $t$. Therein, the value of the weighting function $M_\beta\left(t - \tau\right)$ changes with respect to $\tau$ such that the dynamic characteristic of the memory can be captured. Additionally, by introducing such a weighting function into the economic process, the first derivative of $Y\left(t\right)$, i.e., $\frac{{\rm{d}}}{{\rm{d}}t}Y\left(t\right)=M_\beta\left(t\right)X\left(0\right) + \int_0^t M_\beta\left(t - \tau\right)\left[\frac{{\rm{d}}}{{\rm{d}}\tau} X\left(\tau\right)\right]{{\rm{d}}\tau}$, no longer depends only on $X\left(t\right)$ but also on $X\left(\tau\right)$ with $\tau\in\left[0,t\right)$. The specifical form of $M_\beta\left(t - \tau\right)$ is considered to be in the power form, i.e., $M_\beta\left(t - \tau\right) = \frac{1}{\Gamma\left(\beta\right)} \frac{1}{\left(t-\tau\right)^{1-\beta}}$, where
	\begin{equation}\label{eq:gamma_function}
	\Gamma\left(z\right)=\int^{+\infty}_0 x^{z-1}e^{-x}{{\rm{d}}x}
	\end{equation}
	is the gamma function~\cite{mathworld2019gamma}. It is worth noting that the power-law characteristic of the human's memory has been extensively demonstrated in experiments~\cite{edelman2019evolution}, e.g., power-law forgetting~\cite{wixted1990analyzing} and power-law human learning~\cite{anderson2000learning}.
	
	\item To represent the economic processes in the expression of fractional equation, we take the derivation of $Y\left(t\right)$ at the order of $\beta$ by using the left-sided Caputo fractional derivative. As a result, the memory-aware economic processes can be represented as follows:
	\begin{equation}\label{eq:memory_affecting_time_variant_system}
	{}_{{0}}^{\rm{C}}D_t^\beta Y\left(t\right) = X\left(t\right),
	\end{equation}
	with the initial state for the economic processes, i.e., $Y\left(0\right)=Y_0$, where ${}_{{0}}^{\rm{C}}D_t^\beta Y\left(t\right)$ is the left-sided Caputo fractional derivative of $Y\left(t\right)$ at the order of $\beta$ defined as follows~\cite{tarasova2017logistic}
	\begin{equation}\label{eq:caputo_left_derivative}
	{}_{{0}}^{\rm{C}}D_t^\beta Y\left(t\right)= \frac{1}{\Gamma\left(\left\lceil {\beta} \right\rceil - \beta\right)}\int_{0}^{t} \frac{Y^{\left(\left\lceil {\beta} \right\rceil\right)}\left(\tau\right)}{\left(t-\tau\right)^{\beta+1-\left\lceil {\beta} \right\rceil}} {\rm{d}}\tau,
	\end{equation}
	and $\left\lceil {\beta} \right\rceil$ denotes the integer obtained by rounding up $\beta$. 
\end{enumerate}

As illustrated in~\cite{tarasova2017logistic}, by incorporating the power-law memory, the following key properties of the impact of the user's memory on their decision-making process will be satisfied:
\begin{itemize}
	\item The past experience of the user at different points in time play different roles in its decision-making such that the dynamic characteristic of the memory can be captured.
	
	\item The user is affected by its experience within the memory rather than that at the present point in time. For example, unlike the economic processes studied in~\cite{niyato2008dynamics} and~\cite{semasinghe2014evolutionary} that the users are only aware of their instantaneous achievable utilities, the memory-affecting users are allowed to take into account their past experience within their memory.
	
	\item The memory-affecting users perform to react differently from the memory-unaware users do, which can be verified in reality that the empirical stock trader prefers conservative style rather than radical one. As such, the empirical trader reacts more cautiously than the newbie does.
\end{itemize}

Note that the fractional calculus has been widely adopted to depict the memory effect in multiple disciplines such as economics, e.g.,~\cite{tarasova2016fractional, tarasova2016elasticity}, and human science, e.g.,~\cite{dassiosfractional, lubashevsky2014fractional}. Specifically, the generalization of the economic model of natural growth with power-law memory suggested in~\cite{tarasova2016fractional} applied the fractional differential equation to investigate the proportional relationship between the income and the growth rate of the output.~\cite{dassiosfractional} presented a method to estimate the rate at which the human brain learns a certain amount of knowledge is proportional to the amount of knowledge yet to be learned by using the fractional differential equation.

%\clearpage

\section{System Description}
\label{sec:system_model}

In this section, we give two representative system models. For the purpose to justify the advantage of incorporating the memory effect and avoid being affected by any other environmental factor, we first consider a system model, i.e., HetNet with homogeneous users, in Section~\ref{subsec:system_model_special_case}. 
%First, we present a system model in Section~\ref{subsec:system_model_special_case} to analyze the behaviors of the homogeneous users in the 5G HetNet. 
Then, a generalization of the simplified system model, which considers the scenario with heterogeneous users, is presented in Section~\ref{subsec:system_model_general}.

%\clearpage

\subsection{A Heterogeneous Network with Homogeneous Users}
\label{subsec:system_model_special_case}

We first consider a HetNet, where three BSs, i.e., UHF BS $u$, mmWave BS $m$, and UAV-enabled mmWave BS $a$, are deployed in the network. The homogeneous users are regarded as a group, denoted by ${\cal{N}}^G$, to select the BSs and access the communication services. That is, $N^G_u$, $N^G_m$, and $N^G_a$ users will select UHF BS $u$, mmWave BS $m$, and UAV-enabled mmWave BS $a$, respectively, where $N^G_u+N^G_m+N^G_a=N^G$ is the total number of the users. Accordingly, each user in the group selects UHF BS $u$, mmWave BS $m$, and UAV-enabled mmWave BS $a$ at the probabilities of $y_u=\frac{N^G_u}{N^G}$, $y_m=\frac{N^G_m}{N^G}$, and $y_a=\frac{N^G_a}{N^G}$, respectively. Here, in the system model of the HetNet with homogeneous users, we assume that the users are homogeneous in the long-term basis. For example, the users have the same antenna gain and are likely to be in the same locations, e.g., indoor users working in the office.

\subsubsection{User's Utility}

Based on the expression of the downlink transmission rate of UHF BS $u$, i.e.,~(\ref{eq:UHF_rate}), the expected downlink transmission rate for the users selecting UHF BS $u$ can be derived as follows. As the expected number of the users selecting UHF BS $u$ is $y_{u}N^G$ and UHF BS $u$ is based on TDMA, each user will access the channel with full bandwidth $W_u$ at the probability of $\frac{1}{y_{u}N^G}$ for every time shot. Therefore, the user can obtain the expected bandwidth of $\frac{W_u}{y_{u}N^G}$ in every time slot. Then, the expected downlink transmission rate of UHF BS $u$ is
\begin{equation}\label{eq:rate_UHF_special}
%\small{
{\tilde R}_{u}=\frac{W_u}{y_{u}N^G}\log_2\left(1+ \frac{P_{u}} {\sigma^2_u}\right).
%}
\end{equation}
Let $\lambda_{u}$ denote the intrinsic value of unit bitrate, i.e, the utility that each user can obtain from downloading one unit data by selecting the communication service of UHF BS $u$, the utility excluding the cost that each user can obtain from selecting UHF BS $u$ is therefore $\lambda_u{\tilde{R}}_{u}$. To access the UHF BS $u$'s communication service, the user needs to pay the service fee, which is proportional to the user's download and further the transmission rate. We denote the price of downloading one unit data through UHF BS $u$'s communication service by $\phi_{u}$, the utility including the cost that each user can obtain by selecting UHF BS $u$ is given by
\begin{equation}\label{eq:utility_UHF_special}
%\small{
\Pi_{u}=\lambda_{u}{\tilde{R}}_{u}-\phi_{u}{\tilde{R}}_{u}=\left(\lambda_{u}-\phi_{u}\right){\tilde{R}}_{u}=\lambda^\phi_{u}{\tilde{R}}_{u},
%}
\end{equation}
where $\lambda^\phi_{u} = \lambda_{u}-\phi_{u}$.

Similar to~(\ref{eq:rate_UHF_special}) and based on~(\ref{eq:mmwave_rate}) and~(\ref{eq:uav_rate}), the expected downlink transmission rates of mmWave BS $m$ and UAV-enabled mmWave BS $a$ for each user are ${\tilde{R}}_{m}=\frac{W_m}{y_{m}N^G}\log_2\left(1+\frac{P_{m}} {\sigma_m^2}\right)$
and ${\tilde{R}}_{a}=\frac{W_a}{y_{a}N^G}\log_2\left(1+ \frac{P_{a}} {\sigma_a^2}\right)$, respectively. Then, similar to~(\ref{eq:utility_UHF_special}), the utilities of each user obtained from selecting mmWave BS $m$ and UAV-enabled mmWave BS $a$ are respectively
\begin{equation}\label{eq:utility_mmWave_special}
%\small{
\Pi_{m}=\lambda_{m}{\tilde{R}}_{m}-\phi_{m}{\tilde{R}}_{m}=\left(\lambda_{m}-\phi_{m}\right){\tilde{R}}_{m}=\lambda^\phi_{m}{\tilde{R}}_{m}
%}
\end{equation}
and
\begin{equation}\label{eq:utility_UAV_special}
%\small{
\Pi_{a}=\lambda_{a}{\tilde{R}}_{a}-\phi_{a}{\tilde{R}}_{a}=\left(\lambda_{a}-\phi_{a}\right){\tilde{R}}_{a}=\lambda^\phi_{a}{\tilde{R}}_{a},
%}
\end{equation}
where $\lambda_{m}$ and $\lambda_{a}$ are respectively the intrinsic values of unit bitrate for mmWave BS $m$ and UAV-enabled mmWave BS $a$, $\phi_{m}$ and $\phi_{a}$ are respectively the price of downloading one unit data through the  mmWave BS $m$'s and UAV-enabled mmWave $a$'s communication service,  $\lambda^\phi_{m} = \lambda_{m}-\phi_{m}$, and $\lambda^\phi_{a} = \lambda_{a}-\phi_{a}$.

\subsubsection{Game Formulation for the Heterogeneous Cellular Network with Homogeneous Users}

With the utility functions defined in~(\ref{eq:utility_UHF_special}),~(\ref{eq:utility_mmWave_special}), and~(\ref{eq:utility_UAV_special}), i.e., $\Pi_{u}$, $\Pi_{m}$, and $\Pi_{a}$, respectively, the average utility of the user is 
${\bar{\Pi}}=y_{u}\Pi_{u}+ y_{m}\Pi_{m}+ y_{a}\Pi_{a}$, where $y_{u}$, $y_{m}$, and $y_{a}$ respectively represent the probabilities that the users select UHF BS $u$, mmWave BS $m$, and UAV-enabled mmWave BS $a$. Then, the replicator dynamics yields the following classical evolutionary game, i.e., the system of Ordinary Differential Equations (ODEs):
\begin{equation}\label{eq:evolutionary_game_special}
%\small{
\begin{aligned}
\frac{\rm{d}}{{\rm{d}} t}{ {y}}_{u}\left(t\right)&=\exp\left(-\delta\right)y_{u}\left(t\right)\left[\Pi_u\left(t\right) -{\bar{\Pi}}\left(t\right)\right],\\
\frac{\rm{d}}{{\rm{d}} t}{ {y}}_{m}\left(t\right)&=\exp\left(-\delta\right)y_{m}\left(t\right)\left[\Pi_m\left(t\right) -{\bar{\Pi}}\left(t\right)\right],\\
\frac{\rm{d}}{{\rm{d}} t}{ {y}}_{a}\left(t\right)&=\exp\left(-\delta\right)y_{a}\left(t\right)\left[\Pi_a\left(t\right) -{\bar{\Pi}}\left(t\right)\right],
\end{aligned}
%}
\end{equation}
with the initial strategies, i.e., the Dirichlet boundary condition, of $y_{u}\left(0\right) = y_{u}^0$, $y_{m}\left(0\right) = y_{m}^0$, and $y_{a}\left(0\right) = y_{a}^0$, where $\frac{\rm{d}}{{\rm{d}} t}$ represents the first derivative with respect to $t$ and $\delta$ controls the gain from the rate of strategy adaptation. Moreover, by incorporating the power-law memory, we formulate the fractional evolutionary game based on the left-sided Caputo fractional derivative as follows:
\begin{equation}\label{eq:fractional_evolutionary_game_special}
%\small{
\begin{aligned}
{}_{{0}}^{\rm{C}}{\rm{D}}_t^\beta  {y}_{u}\left(t\right)&=\exp\left(-\delta\right)y_{u}\left(t\right)\left[\Pi_{u}\left(t\right) -{\bar{\Pi}}\left(t\right)\right],\\
{}_{{0}}^{\rm{C}}{\rm{D}}_t^\beta {y}_{m}\left(t\right)&=\exp\left(-\delta\right)y_{m}\left(t\right)\left[\Pi_{m}\left(t\right) -{\bar{\Pi}}\left(t\right)\right]\\
{}_{{0}}^{\rm{C}}{\rm{D}}_t^\beta  {y}_{a}\left(t\right)&=\exp\left(-\delta\right)y_{a}\left(t\right)\left[\Pi_{a}\left(t\right) -{\bar{\Pi}}\left(t\right)\right],
\end{aligned},
%}
\end{equation}
with the initial strategies of $y_{u}\left(0\right) = y_{u}^0$, $y_{m}\left(0\right) = y_{m}^0$, and $y_{a}\left(0\right) = y_{a}^0$, where $\beta\in\left(0,1\right)\cup\left(1,2\right)$ is the order of the Caputo fractional derivative defined in~(\ref{eq:caputo_left_derivative}). Note here that the specific physical meaning of $\beta$ is explained in Section~\ref{subsec:performance_special_case}. Since our fractional evolutionary game is a standard system of fractional differential equations with the Dirichlet boundary condition, the solution of which can be obtained numerically by using the standard algorithm in~\cite{jain1979numerical}.

%\clearpage

\subsection{A Heterogeneous Cellular Network with Heterogeneous Users}
\label{subsec:system_model_general}

In the network, there are three types of BSs deployed in the HetNet, i.e., UHF BSs, mmWave BSs, and UAV-enabled mmWave BSs, each of which contains a set of BSs, i.e., ${\cal{U}}$, ${\cal{M}}$, and ${\cal{A}}$ respectively denote the sets of UHF BSs, mmWave BSs, and UAV-enabled mmWave BSs. A set of users, denoted by ${\cal{N}}$, can select the BSs for their services, e.g., wireless multimedia. Each user $i\in{\cal{N}}$ will select one of the BSs and access the communication service. That is, user $i$ select UHF BS $u$, mmWave BS $m$, and UAV-enabled mmWave BS $a$ at the probabilities of $x_{u,i}$, $x_{m,i}$, and $x_{a,i}$, respectively. Note here that as a generalization of the scenario of the HetNet with homogeneous users, we extend the number of BSs in each type and consider the users to be independent, the parameters of which are different to each other, i.e., nonidentical.

\subsubsection{Users' Utilities}

Based on the equation of the downlink transmission rate for user $i$ at UHF BS $u$, i.e.,~(\ref{eq:UHF_rate}), the users selecting UHF BS $u$ will share the bandwidth of UHF BS $u$. The expected number of the users selecting UHF BS $u$ is $\sum\limits_{j\in{\cal{N}}_u}x_{u,j}$, and hence each user $i\in{\cal{N}}_u$ will access the channel of full bandwidth $W_u$ at the probability of $\frac{1}{{\sum\limits_{j\in{\cal{N}}_u}x_{u,j}}}$ for every time slot, where ${\cal{N}}_u$ is the set of the users in the coverage of UHF BS $u$, and ${\cal{N}}_v$, ${\cal{N}}_{v'}$, ${\cal{N}}_m$, and ${\cal{N}}_a$ are defined similarly to the way that ${\cal{N}}_u$ is defined. Therefore, the user can obtain the expected bandwidth of $\frac{W_u}{{\sum\limits_{j\in{\cal{N}}_u}x_{u,j}}}$ for every time slot. By selecting UHF BS $u$ at the probability of $x_{u,i}$, the expected channel access of user $i$ at UHF BS $u$ is therefore $W_u\frac{x_{u,i}}{\sum\limits_{j\in{\cal{N}}_u}x_{u,j}}$. As for the SINR of user $i$ at UHF BS $u$, we present a simple example, where three UHF BSs, i.e., $u$, $v$, and $v'$, are deployed in the network, and the scenario of more than three UHF BSs can be further investigated in a similar way. In the simple example, UHF BSs $u$, $v$, and $v'$ are using the same channel, i.e., ${\cal{I}}_u=\left\{v,\,v'\right\}$, and hence will incur co-channel interference to each other. The SINR of user $i$ at UHF BS $u$ depends on the BS selection strategies of the users in the coverage of UHF BSs $v$ and $v'$. Accordingly, there are four cases: 
\begin{enumerate}
	\item none of the users $j\in{\cal{N}}_v$ selects UHF BS $v$, and none of the users $j'\in{\cal{N}}_{v'}$ selects UHF BS $v'$;
	
	\item none of the users $j\in{\cal{N}}_v$ selects UHF BS $v$ while there is at least one user $j'\in{\cal{N}}_{v'}$ selecting UHF BS $v'$;
	
	\item there is at least one user $j\in{\cal{N}}_v$ selecting UHF BS $v$ while none of the users $j'\in{\cal{N}}_{v'}$ selects UHF BS $v'$;
	
	\item there is at least one user $j\in{\cal{N}}_v$ selecting UHF BS $v$, and there is at least one users $j'\in{\cal{N}}_{v'}$ selecting UHF BS $v'$. 
\end{enumerate}
For the first case, there is no co-channel interference from UHF BSs $v$ and $v'$ to UHF BS $u$, the SINR of user $i$ at UHF BS $u$ is therefore $\frac{P_{u,i}} {\sigma^2_u}$ with $\prod\limits_{j\in{\cal{N}}_v}\left(1-x_{v,j}\right)\prod\limits_{j'\in{\cal{N}}_{v'}}\left(1-x_{v',j'}\right)$ as the case happening probability\footnote{As we assume in Section~\ref{subsubsec:uhf_propagation} that the UHF BSs use TDMA, user $i\in{\cal{N}}_u$ will suffer from the co-channel interference caused by the communication service of UHF BS $w\in{\cal{I}}_u$ only when there is at least one user $j\in{\cal{N}}_w$ selecting UHF BS $w$.}. Similarly, the SINRs of user $i$ at UHF BS $u$ in the second, third, and fourth cases are $\frac{P_{u,i}} {P_{v',i}+\sigma^2_u}$, $\frac{P_{u,i}} {P_{v,i}+\sigma^2_u}$, and $\frac{P_{u,i}} {\sum\limits_{w\in{\cal{I}}_u}P_{w,i}+\sigma^2_u}$, respectively. Moreover, the happening probabilities for the second, third, and fourth cases are $\prod\limits_{j\in{\cal{N}}_v}\left(1-x_{v,j}\right)\left[1-\prod\limits_{j'\in{\cal{N}}_{v'}}\left(1-x_{v',j'}\right)\right]$, $\left[1-\prod\limits_{j\in{\cal{N}}_v}\left(1-x_{v,j}\right)\right]\prod\limits_{j'\in{\cal{N}}_{v'}}\left(1-x_{v',j'}\right)$, and $\left[1-\prod\limits_{j\in{\cal{N}}_v}\left(1-x_{v,j}\right)\right]\left[1-\prod\limits_{j'\in{\cal{N}}_{v'}}\left(1-x_{v',j'}\right)\right]$, respectively. Based on~(\ref{eq:UHF_rate}), the expected downlink transmission rate of user $i$ at UHF BS $u$ is therefore defined in~(\ref{eq:uhf_rate_general}).

\begin{figure*}
	\begin{equation}\label{eq:uhf_rate_general}
	%\small{
	\begin{aligned}
	{\bar R}_{u,i}=&W_u\frac{x_{u,i}}{\sum\limits_{j\in{\cal{N}}_u}x_{u,j}}\left\{\log_2\left(1+ \frac{P_{u,i}} {\sigma^2_u}\right)\prod\limits_{j\in{\cal{N}}_v}\left(1-x_{v,j}\right)\prod\limits_{j'\in{\cal{N}}_{v'}}\left(1-x_{v',j'}\right)\right.\\
	&+\log_2\left(1+ \frac{P_{u,i}} {P_{v',i}+\sigma^2_u}\right)\prod\limits_{j\in{\cal{N}}_v}\left(1-x_{v,j}\right)\left[1-\prod\limits_{j'\in{\cal{N}}_{v'}}\left(1-x_{v',j'}\right)\right] \\
	&+\log_2\left(1+ \frac{P_{u,i}} {P_{v,i}+\sigma^2_u}\right)\left[1-\prod\limits_{j\in{\cal{N}}_v}\left(1-x_{v,j}\right)\right]\prod\limits_{j'\in{\cal{N}}_{v'}}\left(1-x_{v',j'}\right)\\
	&\left.+\log_2\left(1+ \frac{P_{u,i}} {\sum\limits_{w\in{\cal{I}}_u}P_{w,i}+\sigma^2_u}\right)\left[1-\prod\limits_{j\in{\cal{N}}_v}\left(1-x_{v,j}\right)\right]\left[1-\prod\limits_{j'\in{\cal{N}}_{v'}}\left(1-x_{v',j'}\right)\right]	\right\}.
	\end{aligned}
	%}
	\end{equation}
\end{figure*}

Let $\phi_{u,i}$ denote the price of downloading one unit data through UHF BS $u$'s communication service for user $i$. Then, the utility of user $i$ obtained from ing UHF BS $u$ is a function of the transmission rate given by
\begin{equation}\label{eq:utility_UHF_general}
%\small{
U_{u,i}=\lambda_{u,i}{\bar{R}}_{u,i}-\phi_{u,i}{\bar{R}}_{u,i}=\left(\lambda_{u,i}-\phi_{u,i}\right){\bar{R}}_{u,i}=\lambda^\phi_{u,i}{\bar{R}}_{u,i},
%}
\end{equation}
where $\lambda_{u,i}$ denotes the intrinsic value of unit bitrate for user $i$ and $\lambda^\phi_{u,i} = \lambda_{u,i}-\phi_{u,i}$.

Similar to~(\ref{eq:uhf_rate_general}) and based on~(\ref{eq:mmwave_rate}) and~(\ref{eq:uav_rate}), the expected downlink transmission rates of user $i$ at mmWave BS $m$ and UAV-enabled mmWave BS $a$ are
${\bar{R}}_{m,i}=W_m\frac{x_{m,i}}{\sum\limits_{j\in{\cal{N}}_m}x_{m,j}}\log_2\left(1+\frac{P_{m,i}} {\sigma_m^2}\right)$ and $ {\bar{R}}_{a,i}=W_a\frac{x_{a,i}}{\sum\limits_{j\in{\cal{N}}_a}x_{a,j}}\log_2\left(1+ \frac{P_{a,i}} {\sigma_a^2}\right)$, respectively. Accordingly, the utilities of user $i$ obtained from selecting mmWave BS $m$ and UAV-enabled mmWave BS $a$ are respectively
\begin{equation}\label{eq:utility_mmWave_general}
%\small{
\begin{aligned}
U_{m,i}=&\lambda_{m,i}{\bar{R}}_{m,i}-\phi_{m,i}{\bar{R}}_{m,i}\\
=&\left(\lambda_{m,i}-\phi_{m,i}\right){\bar{R}}_{m,i} = \lambda_{m,i}^\phi{\bar{R}}_{m,i}
\end{aligned}
%}
\end{equation}
and
\begin{equation}\label{eq:utility_UAV_general}
%\small{
U_{a,i}=\lambda_{a,i}{\bar{R}}_{a,i}-\phi_{a,i}{\bar{R}}_{a,i}=\left(\lambda_{a,i}-\phi_{a,i}\right){\bar{R}}_{a,i} = \lambda_{a,i}^\phi{\bar{R}}_{a,i},
%}
\end{equation}
where $\lambda_{m,i}$ and $\lambda_{a,i}$ are respectively the intrinsic values of unit bitrate of mmWave BS $m$ and UAV-enabled mmWave BS $a$ for user $i$, $\phi_{m,i}$ and $\phi_{a,i}$ are respectively the prices of downloading one unit data through mmWave BS $m$'s and UAV-enabled mmWave $a$'s communication services for user $i$, $\lambda_{m,i}^\phi = \lambda_{m,i}-\phi_{m,i}$, and $\lambda_{a,i}^\phi = \lambda_{a,i}-\phi_{a,i}$.

\subsubsection{Game Formulation for the Heterogeneous Cellular Network with Heterogeneous Users}

With the utilities of user $i$ obtained from UHF BS $u$, mmWave BS $m$, and UAV-enabled mmWave BS $a$, i.e.,~(\ref{eq:utility_UHF_general}),~(\ref{eq:utility_mmWave_general}), and~(\ref{eq:utility_UAV_general}), respectively, the average utility of user $i$ from all the BSs is expressed by ${\bar{U}}_i=\sum\limits_{u\in{\cal{C}}^{\rm{U}}_i}x_{u,i}U_{u,i}+ \sum\limits_{m\in{\cal{C}}^{\rm{M}}_i}x_{m,i}U_{m,i}+
\sum\limits_{a\in{\cal{C}}^{\rm{A}}_i}x_{a,i}U_{a,i}$, where ${\cal{C}}^{\rm{U}}_i$, ${\cal{C}}^{\rm{M}}_i$, and ${\cal{C}}^{\rm{A}}_i$ are the sets of UHF BSs, mmWave BSs, and UAV-enabled mmWave BSs covering user $i$, respectively. Then, the replicator dynamics yields the following classical evolutionary game:
\begin{equation}\label{eq:evolutionary_game}
%\small{
\begin{aligned}
\frac{\rm{d}}{{\rm{d}} t}{ {x}}_{u,i}\left(t\right)&=\exp\left(-\delta\right)x_{u,i}\left(t\right)\left[U_{u,i}\left(t\right) -{\bar{U}}_i\left(t\right)\right]\\
\frac{\rm{d}}{{\rm{d}} t}{ {x}}_{m,i}\left(t\right)&=\exp\left(-\delta\right)x_{m,i}\left(t\right)\left[U_{m,i}\left(t\right) -{\bar{U}}_i\left(t\right)\right]\\
\frac{\rm{d}}{{\rm{d}} t}{ {x}}_{a,i}\left(t\right)&=\exp\left(-\delta\right)x_{a,i}\left(t\right)\left[U_{a,i}\left(t\right) -{\bar{U}}_i\left(t\right)\right]
\end{aligned}
%}
\end{equation}
with the initial strategies of $x_{u,i}\left(0\right) = x_{u,i}^0$, $x_{m,i}\left(0\right) = x_{m,i}^0$, and $x_{a,i}\left(0\right) = x_{a,i}^0$. By incorporating the power-law memory, we further formulate the fractional evolutionary game based on the left-sided Caputo fractional derivative as follows:
\begin{equation}\label{eq:fractional_evolutionary_game_general}
%\small{
\begin{aligned}
{}_{{0}}^{\rm{C}}{\rm{D}}_t^\beta  {x}_{u,i}\left(t\right)&=\exp\left(-\delta\right)x_{u,i}\left(t\right)\left[U_{u,i}\left(t\right) -{\bar{U}}_i\left(t\right)\right]\\
{}_{{0}}^{\rm{C}}{\rm{D}}_t^\beta {x}_{m,i}\left(t\right)&=\exp\left(-\delta\right)x_{m,i}\left(t\right)\left[U_{m,i}\left(t\right) -{\bar{U}}_i\left(t\right)\right]\\
{}_{{0}}^{\rm{C}}{\rm{D}}_t^\beta  {x}_{a,i}\left(t\right)&=\exp\left(-\delta\right)x_{a,i}\left(t\right)\left[U_{a,i}\left(t\right) -{\bar{U}}_i\left(t\right)\right]
\end{aligned}
%}
\end{equation}
with the initial strategies of $x_{u,i}\left(0\right) = x_{u,i}^0$, $x_{m,i}\left(0\right) = x_{m,i}^0$, and $x_{a,i}\left(0\right) = x_{a,i}^0$, where $\beta\in\left(0,1\right)\cup\left(1,2\right)$ is the order of the left-sided Caputo fractional derivative defined in~(\ref{eq:caputo_left_derivative}).

%\clearpage

\section{Equilibrium Analysis}
\label{sec:equilibrium_analysis}

In this section, we analyze the equilibrium of the fractional evolutionary game defined in~(\ref{eq:fractional_evolutionary_game_general}), where the existence, uniqueness, and stability of the equilibrium are theoretically verified. We first transform the fractional evolutionary game defined in~(\ref{eq:fractional_evolutionary_game_general}) into an equivalent problem in Theorem~\ref{th:equivalent_problem}. Then, we analytically prove that the equivalent problem admits a unique solution in Theorem~\ref{th:unique_solution_equivalent_problem}, which implies the existence and uniqueness of the equilibrium of the fractional evolutionary game defined in~(\ref{eq:fractional_evolutionary_game_general}). Furthermore, we prove the equilibrium of the fractional evolutionary game defined in~(\ref{eq:fractional_evolutionary_game_general}) to be uniformly stable in Theorem~\ref{th:fractional_equilibrium_stability}.

%\clearpage

\subsection{Existence and Uniqueness on the Fractional Evolutionary Equilibrium}

Considering the fractional evolutionary game defined in~(\ref{eq:fractional_evolutionary_game_general}), let ${\bf{X}}\left(t\right)=\left[x_{w,i}\left(t\right)\right]_{w\in{\cal{U}}\cup{\cal{M}}\cup{\cal{A}}, i\in{\cal{N}}}$ and
${\bf{F}}\left({\bf{X}}\left(t\right)\right)=\left[ \exp\left(-\delta\right) x_{w,i}\left(t\right) \left[U_{w,i}\left(t\right) -{\bar{U}}_i\left(t\right)\right]\right]_{w\in{\cal{U}}\cup{\cal{M}}\cup{\cal{A}}, i\in{\cal{N}}}$, we accordingly have the fractional evolutionary game~(\ref{eq:fractional_evolutionary_game_general}) in a simplified form as follows:
\begin{equation}\label{eq:fractional_evolutionary_game_vector}
%\small{
{}_{{0}}^{\rm{C}}{\rm{D}}_t^\beta  {\bf{X}}\left(t\right) = {\bf{F}}\left({\bf{X}}\left(t\right)\right),
%}
\end{equation}
where $ {\bf{X}}\left(0\right)={\bf{X}}^{\rm{0}}=\left[x_{w,i}^{\rm{0}}\right]_{w\in{\cal{U}}\cup{\cal{M}}\cup{\cal{A}}, i\in{\cal{N}}}$ and $t\in{\cal{T}}=\left[0,T\right]$.

\begin{theorem}\label{th:equivalent_problem}
	If $f_j$, i.e., the $j$-th element of the vector ${\bf{F}}$, $j\in \left\{1, 2, \ldots, N\left(U+A+M\right)\right\}$, satisfies the following conditions:
	\begin{itemize}
		\item $f_j: {\cal{D}} \mapsto {\mathbb{R}}^+$ is second-order continuous,
		
		\item $\frac{\partial}{\partial x_{w,i}}f_j$ exists and is bounded on ${\cal{D}}$, $\forall w\in{\cal{U}}\cup{\cal{M}}\cup{\cal{A}}$, $\forall i\in{\cal{N}}$,
	\end{itemize}
	(\ref{eq:fractional_evolutionary_game_vector}) has an equivalent problem as follows:
	\begin{equation}\label{eq:fractional_evolutionary_game_vector_equivalent}
	%\small{
	{\bf{X}}\left(t\right) = {\bf{X}}^{\rm{0}} + {}_{{0}}{\rm{I}}_t^\beta {\bf{F}}\left({\bf{X}}\left(t\right)\right), \quad \forall t \in {\cal{T}} = \left[0,T\right],
	%}
	\end{equation}
	where ${}_{{0}}{\rm{I}}_t^\beta$ is the fractional integral defined as follows~\cite{podlubny1998fractional}:
	\begin{equation}\label{eq:fractional_integral}
	%\small{
	{}_{{0}}{\rm{I}}_t^\beta f\left(t\right) = \int_{0}^{t}\frac{\left(t-\tau\right)^{\beta-1}}{\Gamma\left(\beta\right)} f\left(\tau\right) {\rm{d}}\tau,
	%}
	\end{equation}
	the second condition implies the Lipschitz condition, i.e., for all $j\in \left\{1, 2, \ldots, N\left(U+A+M\right)\right\}$, $\left|f_j\left({\bf{X}}\left(t\right)\right) - f_j\left({\bf{Y}}\left(t\right)\right)\right| < L \left\| {\bf{X}}\left(t\right) - {\bf{Y}}\left(t\right) \right\|_{{\cal{L}}_1}$, where $L\in{\mathbb{R}}^+$, and $N$, $U$, $A$, and $M$ are the cardinalities of ${\cal{N}}$, ${\cal{U}}$, ${\cal{A}}$, and ${\cal{M}}$, respectively.
\end{theorem}

\begin{proof}
	From~(\ref{eq:fractional_evolutionary_game_vector_equivalent}), we have
	\begin{equation}\label{eq:fractional_integral_integer_derivative}
	%\small{
	\frac{\rm{d}^{\left\lceil {\beta} \right\rceil}}{{\rm{d}}t^{\left\lceil {\beta} \right\rceil}}{\bf{X}}\left(t\right) = \frac{\rm{d}^{\left\lceil {\beta} \right\rceil}}{{{\rm{d}}t}^{\left\lceil {\beta} \right\rceil}}\left[{}_{{0}}{\rm{I}}_t^\beta  {\bf{F}}\left({\bf{X}}\left(t\right)\right)\right] = {}_{{0}}^{\rm{RL}}{\rm{D}}_t^{\left\lceil {\beta} \right\rceil - \beta} {\bf{F}}\left({\bf{X}}\left(t\right)\right),
	%}
	\end{equation}
	where ${}_{{0}}^{\rm{RL}}{\rm{D}}_t^{\alpha} z\left(t\right)$ is the Riemann-Liouville fractional derivative with respect to $t$ defined as follows~\cite{heymans2006physical}:
	\begin{equation}\label{eq:Riemann_Liouville_fractional}
	%\small{
	\begin{aligned}
	&{}_{{0}}^{\rm{RL}}{\rm{D}}_t^{\alpha} z\left(t\right) = \frac{\rm{d}^{\left\lceil {\alpha} \right\rceil}}{{{\rm{d}}t}^{\left\lceil {\alpha} \right\rceil}}\left[{}_{{0}}{\rm{I}}_t^{\left\lceil {\alpha} \right\rceil -\alpha} z\left(t\right)\right] \\
	=& \frac{\rm{d}^{\left\lceil {\alpha} \right\rceil}}{{{\rm{d}}t}^{\left\lceil {\alpha} \right\rceil}}\left[ \frac{1}{\Gamma\left({\left\lceil {\alpha} \right\rceil -\alpha}\right)} \int_{0}^{t} \frac{z\left(\tau\right)}{\left(t-\tau\right)^{1 - {\left\lceil {\alpha} \right\rceil +\alpha}}} {\rm{d}}\tau\right].
	\end{aligned}
	%}
	\end{equation} 
	In~(\ref{eq:fractional_integral_integer_derivative}), according to the definition of the Riemann-Liouville fractional derivative in~(\ref{eq:Riemann_Liouville_fractional}), we have
	\begin{equation}\label{eq:fractional_integral_integer_derivative_bound}
	%\small{
	\begin{aligned}
	&{}_{{0}}^{\rm{RL}}{\rm{D}}_t^{\left\lceil {\beta} \right\rceil - \beta} {\bf{F}}\left({\bf{X}}\left(t\right)\right) \\
	=& \frac{\rm{d}}{{{\rm{d}}t}}\left[ \frac{1}{\Gamma\left(1 - {\left\lceil {\beta} \right\rceil + \beta}\right)} \int_{0}^{t} \frac{{\bf{F}}\left({\bf{X}}\left(\tau\right)\right)}{\left(t-\tau\right)^{\left\lceil {\beta} \right\rceil - \beta}} {\rm{d}}\tau\right]\\
	=& \frac{1}{\Gamma\left(1 - {\left\lceil {\beta} \right\rceil + \beta}\right)} \frac{\rm{d}}{{{\rm{d}}t}} \int_{0}^{t} {\left(t-\tau\right)^{\beta - \left\lceil {\beta} \right\rceil}} {{\bf{F}}\left({\bf{X}}\left(\tau\right)\right)} {\rm{d}}\tau\\
	=& \frac{1}{\Gamma\left(1 - {\left\lceil {\beta} \right\rceil + \beta}\right)} \frac{\rm{d}}{{{\rm{d}}t}} \int_{0}^{t} {\theta^{\beta - \left\lceil {\beta} \right\rceil}} {{\bf{F}}\left({\bf{X}}\left(t - \theta\right)\right)} {\rm{d}}\theta\\
	=& \frac{1}{\Gamma\left(1 - {\left\lceil {\beta} \right\rceil + \beta}\right)} \left[ {t^{\beta - \left\lceil {\beta} \right\rceil}} {{\bf{F}}\left({\bf{X}}^0\right)} + \int_{0}^{t} {\theta^{\beta - \left\lceil {\beta} \right\rceil}} \frac{\rm{d}}{{{\rm{d}}t}} {{\bf{F}}\left({\bf{X}}\left(t - \theta\right)\right)} {\rm{d}}\theta\right]\\
	=& \frac{1}{\Gamma\left(1 - {\left\lceil {\beta} \right\rceil + \beta}\right)} \left[ {t^{\beta - \left\lceil {\beta} \right\rceil}} {{\bf{F}}\left({\bf{X}}^0\right)} +  \int_{0}^{t} {\left(t - \tau\right)^{\beta - \left\lceil {\beta} \right\rceil}} \frac{\rm{d}}{{{\rm{d}}\tau}} {{\bf{F}}\left({\bf{X}}\left(\tau\right)\right)} {\rm{d}}\tau\right]\\
	=& \frac{t^{\beta - \left\lceil {\beta} \right\rceil}}{\Gamma\left(1 - {\left\lceil {\beta} \right\rceil + \beta}\right)} {{\bf{F}}\left({\bf{X}}^0\right)} + {}_{{0}}^{\rm{C}}{\rm{D}}_t^{\left\lceil {\beta} \right\rceil - \beta} {\bf{F}}\left({\bf{X}}\left(t\right)\right),
	\end{aligned}
	%}
	\end{equation}
	where the variable limit integral derivation method is applied. Substituting the results that we obtain in~(\ref{eq:fractional_integral_integer_derivative_bound}) into (\ref{eq:fractional_integral_integer_derivative}), we can have 
	\begin{equation}\label{eq:fractional_integral_integer_derivative_transform}
	%\small{
	\frac{\rm{d}^{\left\lceil {\beta} \right\rceil}}{{\rm{d}}t^{\left\lceil {\beta} \right\rceil}}{\bf{X}}\left(t\right) = \Xi\left(\beta,{\bf{X}}^0\right) + {}_{{0}}{\rm{I}}_t^\beta \left[\frac{\rm{d}^{\left\lceil {\beta} \right\rceil}}{{{\rm{d}}t}^{\left\lceil {\beta} \right\rceil}}   {\bf{F}}\left({\bf{X}}\left(t\right)\right)  \right],
	%}
	\end{equation}
	where
	\begin{equation}\label{eq:xi}
	%\small{
	\Xi\left(\beta,{\bf{X}}^0\right) = \frac{t^{\beta - \left\lceil {\beta} \right\rceil}}{\Gamma\left(1 - {\left\lceil {\beta} \right\rceil + \beta}\right)} {{\bf{F}}\left({\bf{X}}^0\right)}
	%}
	\end{equation}
	represents the first term of the result in~(\ref{eq:fractional_integral_integer_derivative_bound}). Moreover, let $\sigma\in\left(0,t\right)$, there exists an upper bound of $\Xi\left(\beta, {\bf{X}}^0\right)$ in the ${{\cal{L}}^1}$ norm, which can be expressed as follows:
	\begin{equation}\label{eq:xi_bound}
	%\small{
	\begin{aligned}
	&\left\|\Xi\left(\beta, {\bf{X}}^0\right)\right\|_{{\cal{L}}^1} < \left\|\Xi^\sigma\left(\beta, {\bf{X}}^0\right)\right\|_{{\cal{L}}^1}\\
	=&\left\|\frac{\sigma^{\beta - \left\lceil {\beta} \right\rceil}}{\Gamma\left(1 - {\left\lceil {\beta} \right\rceil + \beta}\right)} {{\bf{F}}\left({\bf{X}}^0\right)}\right\|_{{\cal{L}}^1}
	\end{aligned}
	%}
	\end{equation}	
	Then, based on~(\ref{eq:fractional_integral_integer_derivative_transform}) and~(\ref{eq:xi_bound}) as well as the second condition in Theorem~\ref{th:equivalent_problem}, we have the following inequality expression:
	\begin{equation}\label{eq:fractional_integral_norm_inequality}
	%\small{
	\begin{aligned}
	&\left\|\frac{\rm{d}^{\left\lceil {\beta} \right\rceil}}{{\rm{d}}t^{\left\lceil {\beta} \right\rceil}}{\bf{X}}\left(t\right)\right\|_{\cal{T}} \\
	<& \left\|\Xi^\sigma\left(\beta,{\bf{X}}^0\right)\right\|_{{\cal{L}}^1} + \left\|{}_{{0}}{\rm{I}}_t^\beta  \frac{\rm{d}^{\left\lceil {\beta} \right\rceil}}{{\rm{d}}t^{\left\lceil {\beta} \right\rceil}}{\bf{F}}\left({\bf{X}}\left(t\right)\right)\right\|_{\cal{T}}\\
	<& \left\|\Xi^\sigma\left(\beta,{\bf{X}}^0\right)\right\|_{{\cal{L}}^1} + \left\| {}_{{0}}{\rm{I}}_t^\beta  \frac{\rm{d}^{\left\lceil {\beta} \right\rceil}}{{\rm{d}}t^{\left\lceil {\beta} \right\rceil}} {\bf{X}}\left(t\right)\right\|_{\cal{T}}N\left(U+A+M\right) L,
	\end{aligned}
	%}
	\end{equation}
	where $\left\|z\right\|_{\cal{T}} = \int_{\cal{T}} e^{-\mu t} \left\|z\left(t\right)\right\|_{{\cal{L}}^1} {\rm{d}t}$. 
	
	Regarding the norm in the second term of the right-hand side of the inequality equation~(\ref{eq:fractional_integral_norm_inequality}), i.e., $\left\| {}_{{0}}{\rm{I}}_t^\beta  \frac{\rm{d}^{\left\lceil {\beta} \right\rceil}}{{\rm{d}}t^{\left\lceil {\beta} \right\rceil}} {\bf{X}}\left(t\right)\right\|_{\cal{T}}$, we have
	\begin{equation}\label{eq:fractional_integral_norm_upper_bound}
	%\small{
	\begin{aligned}
	&\left\| {}_{{0}}{\rm{I}}_t^\beta  \frac{\rm{d}^{\left\lceil {\beta} \right\rceil}}{{\rm{d}}t^{\left\lceil {\beta} \right\rceil}} {\bf{X}}\left(t\right)\right\|_{\cal{T}} 
	= \int_{0}^{T} e^{-\mu t} \left\|{}_{{0}}{\rm{I}}_t^\beta  \frac{\rm{d}^{\left\lceil {\beta} \right\rceil}}{{\rm{d}}t^{\left\lceil {\beta} \right\rceil}} {\bf{X}}\left(t\right)\right\| {\rm{d}t}\\
	\le& \int_{0}^{T} e^{-\mu t} \int_{{0}}^t\frac{1}{\Gamma\left(\beta\right)}  \frac{\left\|\frac{\rm{d}^{\left\lceil {\beta} \right\rceil}}{{\rm{d}}s^{\left\lceil {\beta} \right\rceil}} {\bf{X}}\left(s\right)\right\|}{\left(t-s\right)^{1-\beta}} {\rm{d}s}\, {\rm{d}t}\\
	=& \int_{0}^{T} \frac{e^{-\mu s}}{\Gamma\left(\beta\right)} \left\|\frac{\rm{d}^{\left\lceil {\beta} \right\rceil}}{{\rm{d}}s^{\left\lceil {\beta} \right\rceil}} {\bf{X}}\left(s\right)\right\| \int_{{s}}^T  e^{-\mu \left(t-s\right)}\left(t-s\right)^{\beta-1} {\rm{d}t}\, {\rm{d}s}\\
	=& \int_{0}^{T} \frac{e^{-\mu s}}{\Gamma\left(\beta\right)} \left\|\frac{\rm{d}^{\left\lceil {\beta} \right\rceil}}{{\rm{d}}s^{\left\lceil {\beta} \right\rceil}} {\bf{X}}\left(s\right)\right\|  \int_{{0}}^{T-s}  e^{-\mu \theta}\theta^{\beta-1} {\rm{d}\theta}\,{\rm{d}s}\\
	=& \int_{0}^{T} \frac{e^{-\mu s}}{\Gamma\left(\beta\right)} \left\|\frac{\rm{d}^{\left\lceil {\beta} \right\rceil}}{{\rm{d}}s^{\left\lceil {\beta} \right\rceil}} {\bf{X}}\left(s\right)\right\|  \int_{{0}}^{\mu\left(T-s\right)}  e^{-\psi}\left(\frac{\psi}{\mu}\right)^{\beta-1} {\rm{d}\left(\frac{\psi}{\mu}\right)}\,{\rm{d}s}\\
	<& \frac{1}{\mu^\beta\Gamma\left(\beta\right)}\int_{0}^{T} e^{-\mu s} \left\|\frac{\rm{d}^{\left\lceil {\beta} \right\rceil}}{{\rm{d}}s^{\left\lceil {\beta} \right\rceil}} {\bf{X}}\left(s\right)\right\|  {\rm{d}s} \underbrace { \int_{{0}}^{+\infty}  e^{-\psi}\psi^{\beta-1} {\rm{d}\psi}}_{\text{gamma function defined in~(\ref{eq:gamma_function})}}\\
	=& \frac{1}{{\mu}^{\beta}}\left\| \frac{\rm{d}^{\left\lceil {\beta} \right\rceil}}{{\rm{d}}t^{\left\lceil {\beta} \right\rceil}} {\bf{X}}\left(t\right)\right\|_{\cal{T}}.\\
	\end{aligned}
	%}
	\end{equation} 
	Then, we substitute the results obtained from~(\ref{eq:fractional_integral_norm_upper_bound}) into~(\ref{eq:fractional_integral_norm_inequality}) and have 
	\begin{equation}\label{eq:fractional_integral_norm_bound}
	%\small{
	\begin{aligned}
	&\left\|\frac{\rm{d}^{\left\lceil {\beta} \right\rceil}}{{\rm{d}}t^{\left\lceil {\beta} \right\rceil}}{\bf{X}}\left(t\right)\right\|_{\cal{T}} \\
	<& \left\|\Xi^\sigma\left(\beta,{\bf{X}}^0\right)\right\|_{{\cal{L}}^1} + \frac{1}{{\mu}^{\beta}}\left\| \frac{\rm{d}^{\left\lceil {\beta} \right\rceil}}{{\rm{d}}t^{\left\lceil {\beta} \right\rceil}} {\bf{X}}\left(t\right)\right\|_{\cal{T}}N\left(U+A+M\right) L,\\
	\Leftrightarrow &\left\|\frac{\rm{d}^{\left\lceil {\beta} \right\rceil}}{{\rm{d}}t^{\left\lceil {\beta} \right\rceil}}{\bf{X}}\left(t\right)\right\|_{\cal{T}} 
	< \frac{1}{1-\frac{N\left(U+A+M\right) L}{{\mu}^{\beta}}}
	\left\|\Xi^\sigma\left(\beta,{\bf{X}}^0\right)\right\|_{{\cal{L}}^1},
	\end{aligned}
	%}
	\end{equation}
	which means that $\left\|\frac{\rm{d}^{\left\lceil {\beta} \right\rceil}}{{\rm{d}}t^{\left\lceil {\beta} \right\rceil}}{\bf{X}}\left(t\right)\right\|_{\cal{T}} $ is bounded when we choose $\mu$ such that ${N\left(U+A+M\right) L}<{{\mu}^{\beta}}$.
	
	We take the fractional derivative of~(\ref{eq:fractional_evolutionary_game_vector_equivalent}) at the order of $\beta$ as follows:
	\begin{equation}%\small{
	\begin{aligned}
	&{}_{{0}}^{\rm{C}}{\rm{D}}_t^\beta{\bf{X}}\left(t\right) =  {}_{{0}}{\rm{I}}_t^{\left\lceil {\beta} \right\rceil-\beta} \frac{\rm{d}^{\left\lceil {\beta} \right\rceil}}{{\rm{d}}t^{\left\lceil {\beta} \right\rceil}}{\bf{X}}\left(t\right)\\
	\mathop = \limits^{(\ref{eq:fractional_integral_integer_derivative_transform})}& {}_{{0}}{\rm{I}}_t^{\left\lceil {\beta} \right\rceil-\beta} \left\{\Xi\left(\beta,{\bf{X}}^0\right) + {}_{{0}}{\rm{I}}_t^\beta \left[\frac{\rm{d}^{\left\lceil {\beta} \right\rceil}}{{{\rm{d}}t}^{\left\lceil {\beta} \right\rceil}}   {\bf{F}}\left({\bf{X}}\left(t\right)\right)  \right]\right\} = {\bf{F}}\left({\bf{X}}\left(t\right)\right),
	\end{aligned}
	%}
	\end{equation}
	which means that~(\ref{eq:fractional_evolutionary_game_vector}) has the equivalent problem defined in~(\ref{eq:fractional_evolutionary_game_vector_equivalent}), and this completes the proof.	
\end{proof}

Then, we prove that the equivalent problem defined in~(\ref{eq:fractional_evolutionary_game_vector_equivalent}) admits a unique solution, which implies that~(\ref{eq:fractional_evolutionary_game_vector}) also admits a unique solution.

\begin{theorem}\label{th:unique_solution_equivalent_problem}
	With the conditions presented in Theorem~\ref{th:equivalent_problem}, the problem defined in~(\ref{eq:fractional_evolutionary_game_vector_equivalent}) admits a unique solution.
\end{theorem} 

\begin{proof}
	Given an operator $\Lambda:{\cal{D}}_{\bf{X}}\mapsto{\cal{D}}_{\bf{X}}$, we have the inequality expression in~(\ref{eq:equivalent_problem_inditical_operator}).
	\begin{figure*}
	\begin{equation}\label{eq:equivalent_problem_inditical_operator}
	\small{
	\begin{aligned}
	&\left\| \Lambda {\bf{X}}\left(t\right) - \Lambda {\bf{Y}}\left(t\right)\right\|_{\cal{T}} \\
	= &\int_{0}^{T} e^{-\mu t} \left\|{}_{{0}}{\rm{I}}_t^\beta {\bf{F}}\left({\bf{X}}\left(t\right)\right) - {}_{{0}}{\rm{I}}_t^\beta {\bf{F}}\left({\bf{Y}}\left(t\right)\right)\right\| {\rm{d}t} < N\left(U+M+A\right)L\left[\int_{0}^{T} e^{-\mu t} \left\|{}_{{0}}{\rm{I}}_t^\beta {\bf{X}}\left(t\right) - {}_{{0}}{\rm{I}}_t^\beta {\bf{Y}}\left(t\right)\right\| {\rm{d}t}\right]\\
	\le& N\left(U+M+A\right)L\left[ \int_{0}^{T} e^{-\mu t} \int_{{0}}^t\frac{1}{\Gamma\left(\beta\right)}  \frac{\left\| {\bf{X}}\left(s\right) -  {\bf{Y}}\left(s\right)\right\|}{\left(t-s\right)^{1-\beta}} {\rm{d}s}\, {\rm{d}t} \right]\\
	=& \frac{N\left(U+M+A\right)L} {\Gamma\left(\beta\right)}\left[ \int_{0}^T  \int_{{s}}^T e^{-\mu t} \frac{\left\| {\bf{X}}\left(s\right) - {\bf{Y}}\left(s\right)\right\|}{\left(t-s\right)^{1-\beta}} {\rm{d}t}\, {\rm{d}s}\right]\\
	=& \frac{N\left(U+M+A\right)L}{\Gamma\left(\beta\right)}\left[\int_{0}^T  e^{-\mu s} \left\| {\bf{X}}\left(s\right) - {\bf{Y}}\left(s\right)\right\|\int_{{s}}^T  \frac{e^{-\mu \left(t-s\right)}}{\left(t-s\right)^{1-\beta}} {\rm{d}t}\, {\rm{d}s}\right]\\
	=& \frac{N\left(U+M+A\right)L}{\Gamma\left(\beta\right)}\left[ \int_{0}^T  e^{-\mu s} \left\| {\bf{X}}\left(s\right) - {\bf{Y}}\left(s\right)\right\| \int_{{0}}^{T-s}  e^{-\mu \theta}\theta^{\beta-1} {\rm{d}\theta}\, {\rm{d}s} \right]\\
	=& \frac{N\left(U+M+A\right)L}{\Gamma\left(\beta\right)}\left[ \int_{0}^T  e^{-\mu s} \left\| {\bf{X}}\left(s\right) - {\bf{Y}}\left(s\right)\right\| \int_{{0}}^{\mu\left(T-s\right)}  e^{-\psi}\left(\frac{\psi}{\mu}\right)^{\beta-1} {\rm{d}\left(\frac{\psi}{\mu}\right)}\, {\rm{d}s} \right]\\
	=& \frac{N\left(U+M+A\right)L}{{\mu}^{\beta}\Gamma\left(\beta\right)}\left[ \int_{0}^T  e^{-\mu s} \left\| {\bf{X}}\left(s\right) - {\bf{Y}}\left(s\right)\right\| \int_{{0}}^{\mu\left(T-s\right)}  e^{-\psi}\psi^{\beta-1} {\rm{d}\psi}\, {\rm{d}s} \right]\\
	< & \frac{N\left(U+M+A\right)L}{{\mu}^{\beta}\Gamma\left(\beta\right)}\left\| {\bf{X}}\left(t\right) - {\bf{Y}}\left(t\right)\right\|_{\cal{T}} \underbrace {\int_0^{ + \infty } {{e^{ - \sigma }}} {\sigma ^{\beta  - 1}}{\rm{d}}\sigma }_{\text{gamma function defined in~(\ref{eq:gamma_function})}}
	=  \frac{N\left(U+M+A\right)L}{{\mu}^{\beta}}\left\| {\bf{X}}\left(t\right) - {\bf{Y}}\left(t\right)\right\|_{\cal{T}}.
	\end{aligned}
	}
	\end{equation}
    \end{figure*}
	From~(\ref{eq:equivalent_problem_inditical_operator}), we can conclude that $\left\| \Lambda {\bf{X}} - \Lambda {\bf{Y}}\right\|_{\cal{T}}<\left\|  {\bf{X}} - {\bf{Y}}\right\|_{\cal{T}}$ if we choose $\mu$ such that ${{\mu}^{\beta}}>{N\left(U+M+A\right)L}$. Hence, the operator $\Lambda$ has a unique fixed point, which indicates that the equivalent problem defined in~(\ref{eq:fractional_evolutionary_game_vector_equivalent}) has a unique solution, and consequently the equilibrium of the fractional evolutionary game defined in~(\ref{eq:fractional_evolutionary_game_vector}) exists and is unique. The proof is completed.
\end{proof}

%\clearpage
\begin{table*}[!]
	\centering
	\caption{Parameters for the Network Model}
	\begin{tabular}{|c|l||c|l||}
		\hline
		\hline
		{\bf{Symbol}} & {\bf{Typical Value}}&{\bf{Symbol}} & {\bf{Typical Value}}\\
		\hline
		$W_m$, $W_u$ $W_a$ & $1$ GHz, $20$ MHz, $1$ GHz	&$f_m$, $f_u$ & $70$ GHz, $1.8$ GHz (e.g., 4G/LTE)\\
		\hline
		$P_u$, $P_m$, $P_a$ &  $46$ dBm~\cite{singh2015tractable}, $30$ dBm~\cite{singh2015tractable}, $23$ dBm~\cite{gapeyenko2018flexible}&$G_u$, $G_{m,i}$, $G_{a,i}$ & $0$ dBi, ($G^M$, $G^S$), $G^A$\\
		\hline
		$\sigma^2_u$ & $-174{\text{dBm/Hz}} + 10\log_{10}\left(W_u\right)+10{\text{dB}}$~\cite{singh2015tractable} & $G^M$, $G^S$, $G^A$ & $18$ dBi, $-2$ dBi, $18$ dBi~\cite{singh2015tractable}\\
		\hline
		$\sigma^2_m$ &  $-174{\text{dBm/Hz}} + 10\log_{10}\left(W_m\right)+10{\text{dB}}$~\cite{singh2015tractable} & $C$, $D$ & $0.081$, $250$ m~\cite{wang2016physical}\\
		\hline
		$\sigma^2_a$ &  $-174{\text{dBm/Hz}} + 10\log_{10}\left(W_a\right)+10{\text{dB}}$~\cite{singh2015tractable} & $H_a$  &  $20$m~\cite{gapeyenko2018flexible}\\
		\hline
		$\alpha_m^{\text{LOS}}$, $\alpha_m^{\text{NLOS}}$, $\alpha_u$ & $2$, $4$~\cite{ghosh2014millimeter}, $2.7$ & $\theta^S_a$ & $45^\circ $\\
		\hline
		$h_{u,i}$, $h_{m,i}$, $h_{a,i}$ & $\exp\left(1\right)$ &
		$b$, $c$ & $1.5$, $1$	\\
		\hline
	\end{tabular}
	\label{tab:network_parameter}
\end{table*}
\subsection{Stability on the Fractional Evolutionary Equilibrium}

After we verify the existence and uniqueness of the equilibrium of the fractional evolutionary game defined in~(\ref{eq:fractional_evolutionary_game_vector}), we investigate its stability, i.e., robustness, in Theorem~\ref{th:fractional_equilibrium_stability} through studying the variations incurred by the small perturbations of the boundary condition.

\begin{theorem}\label{th:fractional_equilibrium_stability}
	With conditions presented in Theorem~\ref{th:equivalent_problem}, the equilibrium of the fractional evolutionary game defined in~(\ref{eq:fractional_evolutionary_game_vector}) is uniformly stable. 
\end{theorem}

\begin{proof}
	Let ${\widetilde{\bf{X}}}\left(t\right)$ and ${\hat{\bf{X}}}\left(t\right)$ be the solutions of ${}_{{0}}^{\rm{C}}{\rm{D}}_t^\beta  {\bf{X}}\left(t\right) = {\bf{F}}\left({\bf{X}}\left(t\right)\right)$ with the initial strategies of ${\widetilde{\bf{X}}}\left(0\right)={\widetilde{{\bf{X}}}}^{\rm{0}}=\left[{\widetilde{x}}_{w,i}^{\rm{0}}\right]_{w\in{\cal{U}}\cup{\cal{M}}\cup{\cal{A}}, i\in{\cal{N}}}$ and ${\hat{\bf{X}}}\left(0\right)={\hat{{\bf{X}}}}^{\rm{0}}=\left[{\hat{x}}_{w,i}^{\rm{0}}\right]_{w\in{\cal{U}}\cup{\cal{M}}\cup{\cal{A}}, i\in{\cal{N}}}$,
	respectively. 
	Then, similar to Theorems~\ref{th:equivalent_problem} and~\ref{th:unique_solution_equivalent_problem}, we can have the following inequality expression
	\begin{equation}\label{eq:stability_fractional_equilibrium}
	%\small{
	\begin{aligned}
	&\left\| {\widetilde{\bf{X}}}\left(t\right) - {\hat{\bf{X}}}\left(t\right)\right\|_{\cal{T}} \\
	<& \left\|{\widetilde{{\bf{X}}}}^0 - {\hat{{\bf{X}}}}^0\right\|_{{\cal{L}}_1} + \frac{N\left(U+M+A\right)L}{{\mu}^{\beta}}\left\| {\widetilde{\bf{X}}}\left(t\right) - {\hat{\bf{X}}}\left(t\right)\right\|_{\cal{T}}\\
	\Leftrightarrow&\left\| {\widetilde{\bf{X}}}\left(t\right) - {\hat{\bf{X}}}\left(t\right)\right\|_{\cal{T}} < \frac{1}{1-\frac{N\left(U+M+A\right)L}{{\mu}^{\beta}}} \left\|{\widetilde{{\bf{X}}}}^0 - {\hat{{\bf{X}}}}^0\right\|_{{\cal{L}}_1}.
	\end{aligned}
	%}
	\end{equation}
	From~(\ref{eq:stability_fractional_equilibrium}), we can have $\left\| {\widetilde{\bf{X}}}\left(t\right) - {\hat{\bf{X}}}\left(t\right)\right\|_{\cal{T}} < \left\|{\widetilde{{\bf{X}}}}^0 - {\hat{{\bf{X}}}}^0\right\|_{{\cal{L}}_1}$ when we choose $\mu$ such that ${{\mu}^{\beta}}<{N\left(U+M+A\right)L}$, which indicates the uniform stability of the equilibrium of the fractional evolutionary game defined in~(\ref{eq:fractional_evolutionary_game_vector})~\cite{el2007stability}. The proof is completed.
\end{proof}

%\clearpage

\section{Performance Evaluation}
\label{sec:performance}

In this section, we conduct extensive numerical simulations to investigate the behaviors of the users in the classical evolutionary game and that in the fractional evolutionary game. The parameter setting for the network model of the HetNet is shown in Table~\ref{tab:network_parameter}, which is similar to those in~\cite{singh2015tractable} and~\cite{gapeyenko2018flexible}. The domains of definition for $y_u\left(t\right)$, $y_m\left(t\right)$, $y_a\left(t\right)$, $x_{u,i}\left(t\right)$, $x_{m,i}\left(t\right)$, $x_{a,i}\left(t\right)$ are $\left[0,1\right]$ for all $t\in\left[0, T\right]$. We first evaluate the performance of the HetNet with homogeneous users formulated in Section~\ref{subsec:system_model_special_case} and investigate the stability of the equilibriums for the classical evolutionary game and fractional evolutionary game. Then, we investigate the behaviors of the heterogeneous users in the HetNet with small-scale fading under different game formulations, i.e., the classical evolutionary game and fractional evolutionary game.

\begin{table}[!]
	\centering
	\caption{Parameters setting for the heterogeneous network with homogeneous users}
	\label{tab:BS_location_special}
	\begin{tabular}{|c|c||c|c||}
		\hline
		\hline
		{\bf{BSs and user}} & {\bf{Coordinate ($Km$)}} & {\bf{Parameters}}& {\bf{Value}} \\
		\hline
		UHF BS & $\left[1\ 0\ 0\right]$ & $\lambda_{u}^\phi$& $10^{-7}$\\
		\hline
		mmWave BS & $\left[0 \ 0 \ 0\right]$ & $\lambda_{m}^\phi$& $1.5\times10^{-9}$\\
		\hline
		\makecell[l]{UAV-enabled \\mmWave BS}& $\left[0\ 0.1\ 0.02\right]$ & $\lambda_{a}^\phi$ & $10^{-9}$\\
		\hline
		User group& $\left[0\ 0.1\ 0\right]$& ${\cal{N}}^G$, $\delta$ & $10$, $2$\\
		\hline
	\end{tabular}
\end{table}

%\clearpage

\begin{figure}[!]
	\centering
	\begin{minipage}{6cm}
		\centering
		\includegraphics[width=1\textwidth,trim=5 0 5 5,clip]{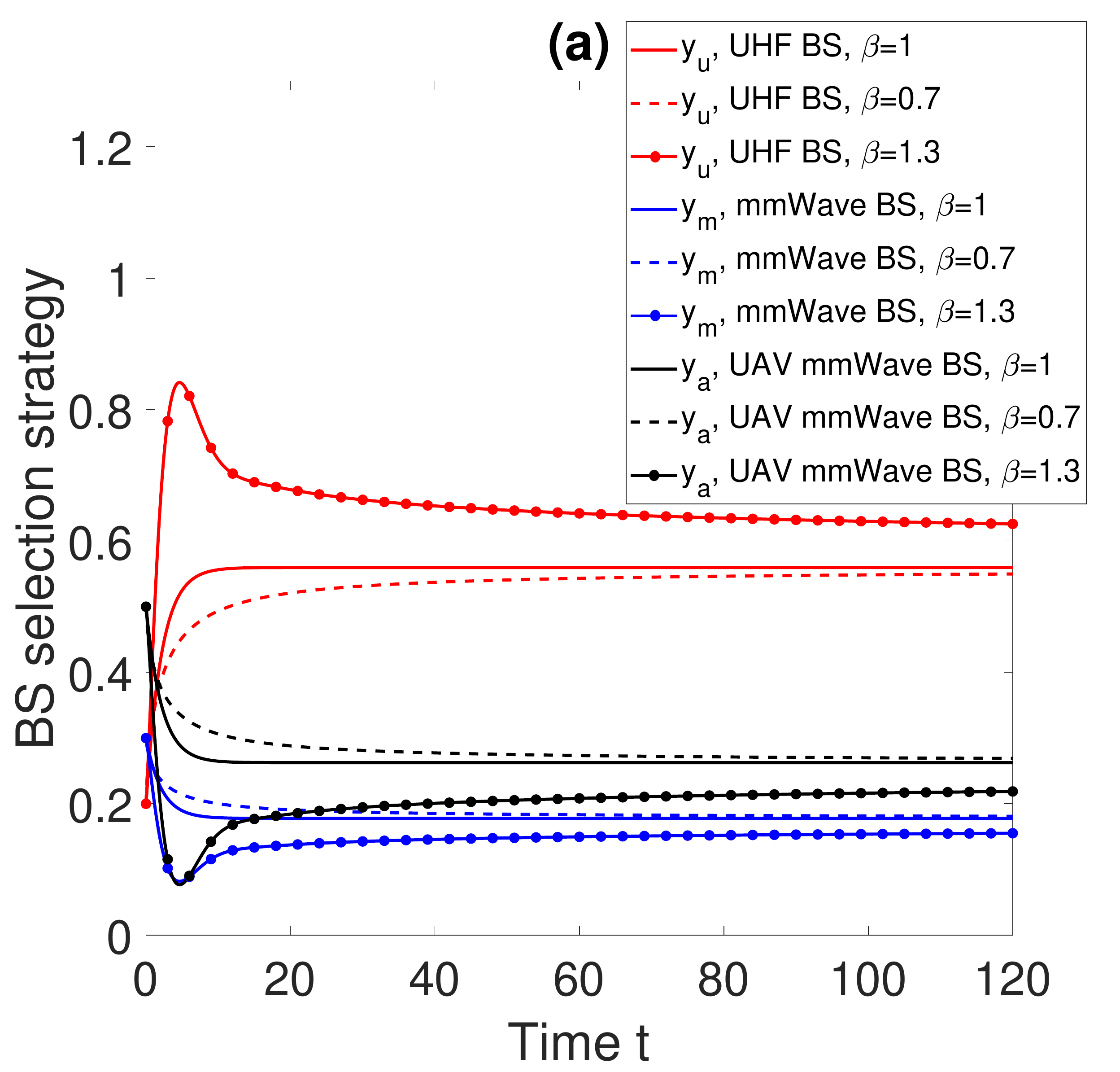}
	\end{minipage}
	\begin{minipage}{6cm}
		\centering
		\includegraphics[width=1\textwidth,trim=5 0 5 5,clip]{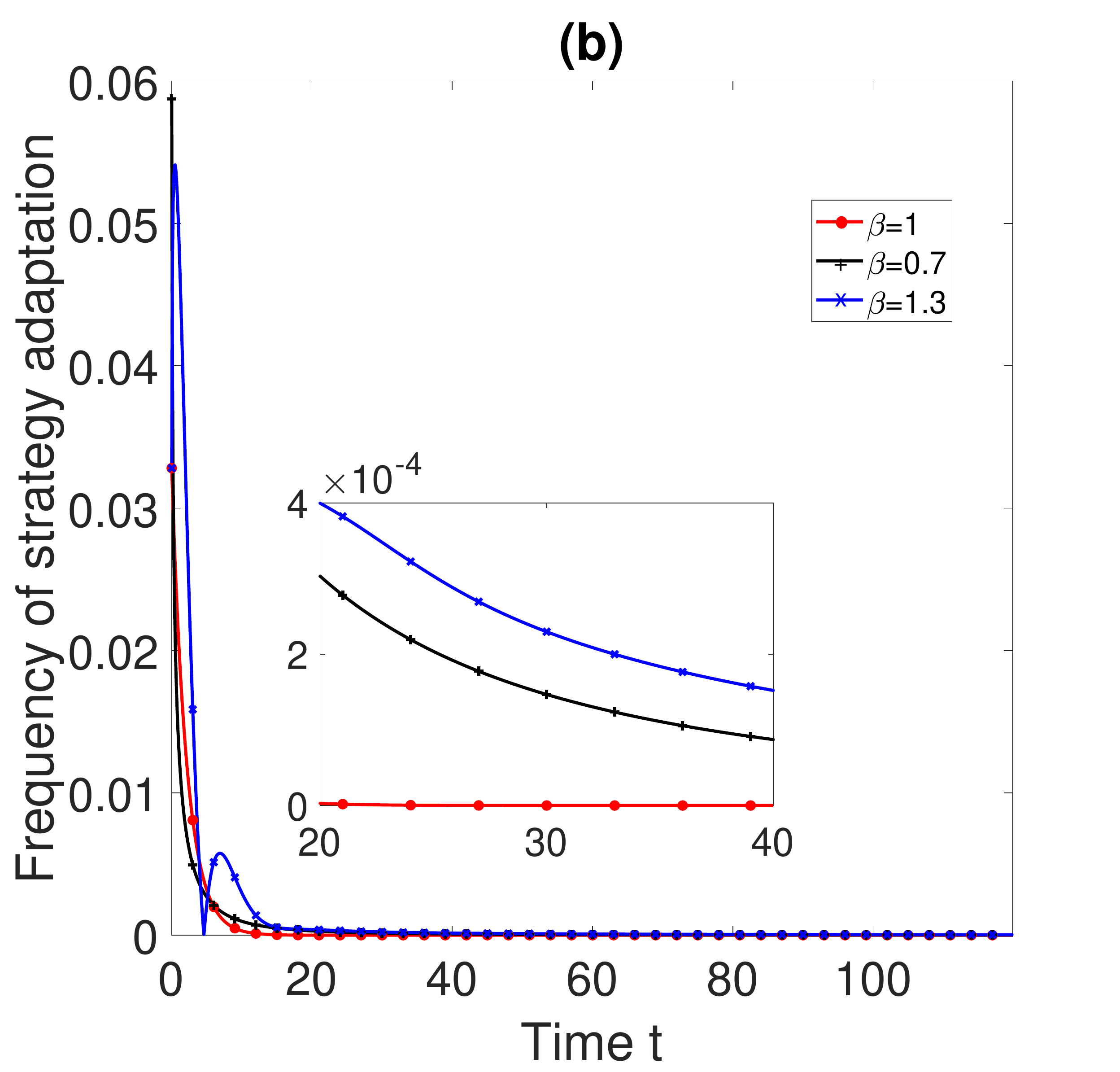}
	\end{minipage}
	\caption{(a) The user's network selection strategy and (b) frequency of strategy adaptation in the classical evolutionary game and fractional evolutionary game}
	\label{fig:special_case_ESS}
\end{figure}

\subsection{Numerical Results for the System Model of the Heterogeneous Network with Homogeneous Users}
\label{subsec:performance_special_case}

We first evaluate the system model formulated in Section~\ref{subsec:system_model_special_case} with the parameter setting given in Table~\ref{tab:BS_location_special}. We compare the results obtained from the fractional evolutionary games with $\beta = 0.7$ and $\beta = 1.3$ with that obtained from the classical evolutionary game, i.e., $\beta = 1$, in Fig.~\ref{fig:special_case_ESS}. As shown in Fig.~\ref{fig:special_case_ESS}(a), the user's strategy of the fractional evolutionary game with $\beta = 1.3$ initially fluctuates in a range more widely than that of the classical evolutionary game and the fractional evolutionary game with $\beta = 0.7$. This means that the strategy adaptation rate of the users in the fractional evolutionary game with $\beta = 1.3$ is faster than that in both the classical evolutionary game and the fractional evolutionary game with $\beta = 0.7$. This is consistent with the result in Fig.~\ref{fig:special_case_ESS}(b) that the frequency of strategy adaptation of the fractional evolutionary game with $\beta = 1.3$ at around $t=4$ is higher than that in both the classical evolutionary game and the fractional evolutionary game with $\beta = 0.7$. However, it can be observed in Fig.~\ref{fig:special_case_ESS}(b) that the user's strategies in the fractional evolutionary game with $\beta = 1.3$ converge to the equilibrium more slowly than that in both the classical evolutionary game and the fractional evolutionary game with $\beta = 0.7$. This is due to the fluctuation in the user's strategy of the fractional evolutionary game with $\beta = 1.3$.

%\clearpage

\begin{figure}[!]
	\centering
	\begin{minipage}{5cm}
		\centering
		\includegraphics[width=1.05\textwidth,trim=5 5 15 5,clip]{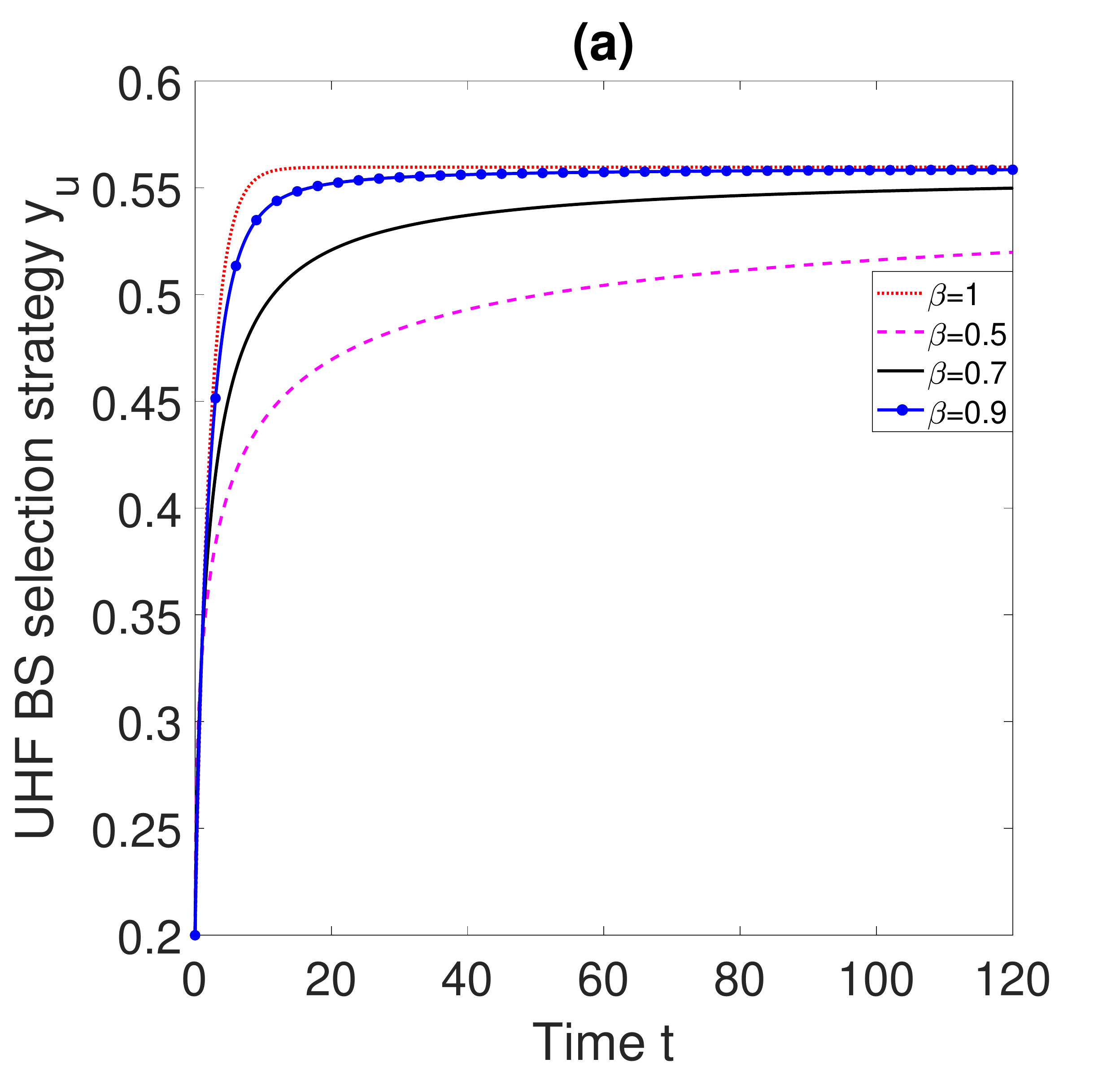}
	\end{minipage}
	\begin{minipage}{5cm}
		\centering
		\includegraphics[width=1.05\textwidth,trim=5 5 15 5,clip]{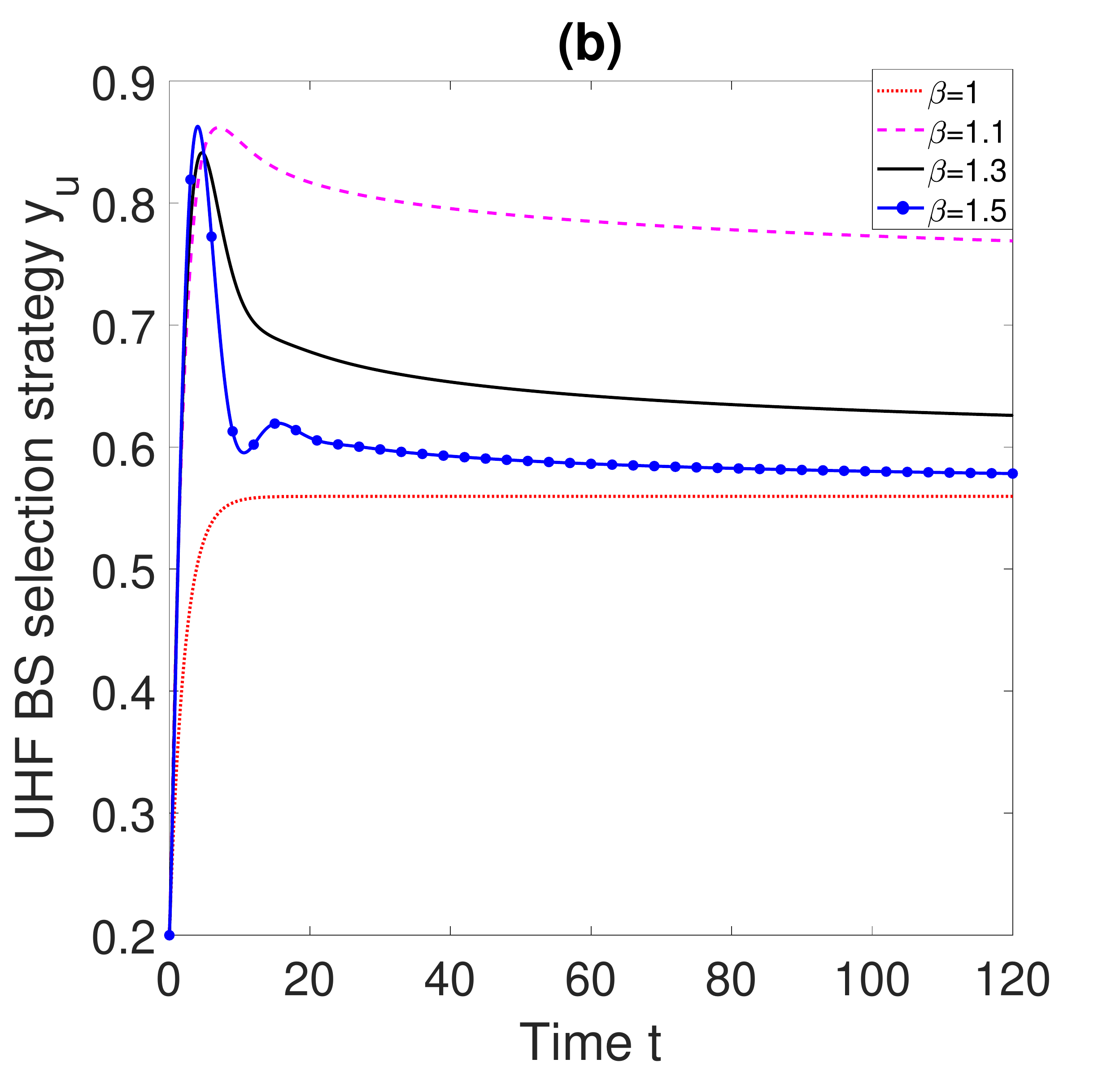}
	\end{minipage}
	\begin{minipage}{5cm}
		\centering
		\includegraphics[width=1.05\textwidth,trim=5 0 15 5,clip]{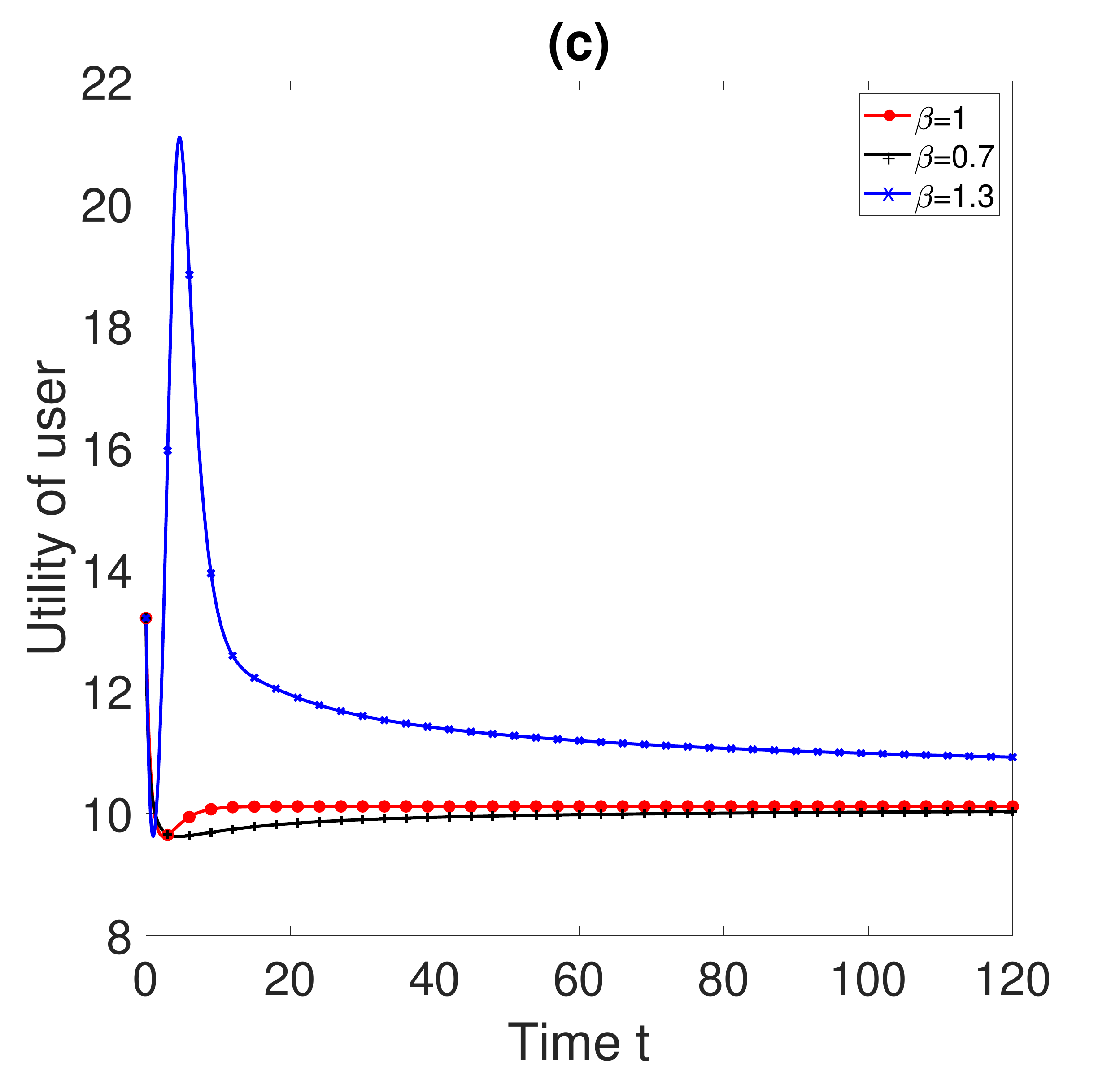}
	\end{minipage}
	\caption{The impact of the coefficient $\beta$ of the left-sided Caputo fractional derivative on the fractional evolutionary game}
	\label{fig:special_case_coefficient}
\end{figure}

To specifically illustrate the impact of the coefficient $\beta$ of the left-sided Caputo fractional derivative on the fractional evolutionary game and further explain its physical meaning, we vary the value of $\beta$ and show the numerical results in Fig.~\ref{fig:special_case_coefficient}. In the figure, the user's strategy fluctuate more widely as the value of $\beta$ increases in the fractional evolutionary games with both $\beta<1$ and $\beta>1$. This means that the strategy adaptation rate increases as the value of $\beta$ increases. Moreover, from Fig.~\ref{fig:special_case_coefficient}(b), as the value of $\beta$ increases, the rate that the user's strategy converges to the vicinity of the equilibrium strategy becomes faster. This is also due to the strategy adaptation rate, which is consistent with the results in and the explanation for Fig.~\ref{fig:special_case_ESS}. 

To specifically illustrate the impact of the coefficient $\beta$ of the left-sided Caputo fractional derivative on the fractional evolutionary game and further explain its physical meaning, we vary the value of $\beta$ and show the results in Fig.~\ref{fig:special_case_coefficient}. In the figure, the rate that the user's strategy converges to the equilibrium becomes faster as the value of $\beta$ increases in the fractional evolutionary game with both $\beta<1$ and $\beta>1$. This means that the convergence rate of the replicator dynamics increases as the value of $\beta$ increases. Moreover, from Fig.~\ref{fig:special_case_coefficient}(b), the rate that the user's strategy converges to the vicinity of the equilibrium strategy becomes faster as the value of $\beta$ increases. This implies that the user's strategy adaptation rate increases as the value of $\beta$ increases, which is consistent with the results in and the explanation for Fig.~\ref{fig:special_case_ESS}. 

In addition, we also evaluate the user's utility under different memory effects, i.e., different values of $\beta$. As shown in Fig.~\ref{fig:special_case_coefficient}(c), the memory effect with $\beta=0.7\in\left(0,1\right)$ can lead to a worse utility for the user compared with the model without memory effect, i.e., $\beta = 1$. In contrast, the memory effect with $\beta = 1.3 \in\left(1,2\right)$ can lead to a better utility for the user compared with the model without memory effect. The results observed in Fig.~\ref{fig:special_case_coefficient}(c) are consistent with the results in~\cite{tarasova2016fractional}. In this case, we would like to define the memory effect with $\beta\in\left(0,1\right)$ as negative memory effect and that with $\beta\in\left(1,2\right)$ as positive memory effect.

%\clearpage

\begin{figure*}[!]
	\centering
	\begin{minipage}{5cm}
		\centering
		\includegraphics[width=1.05\textwidth,trim=5 0 10 5,clip]{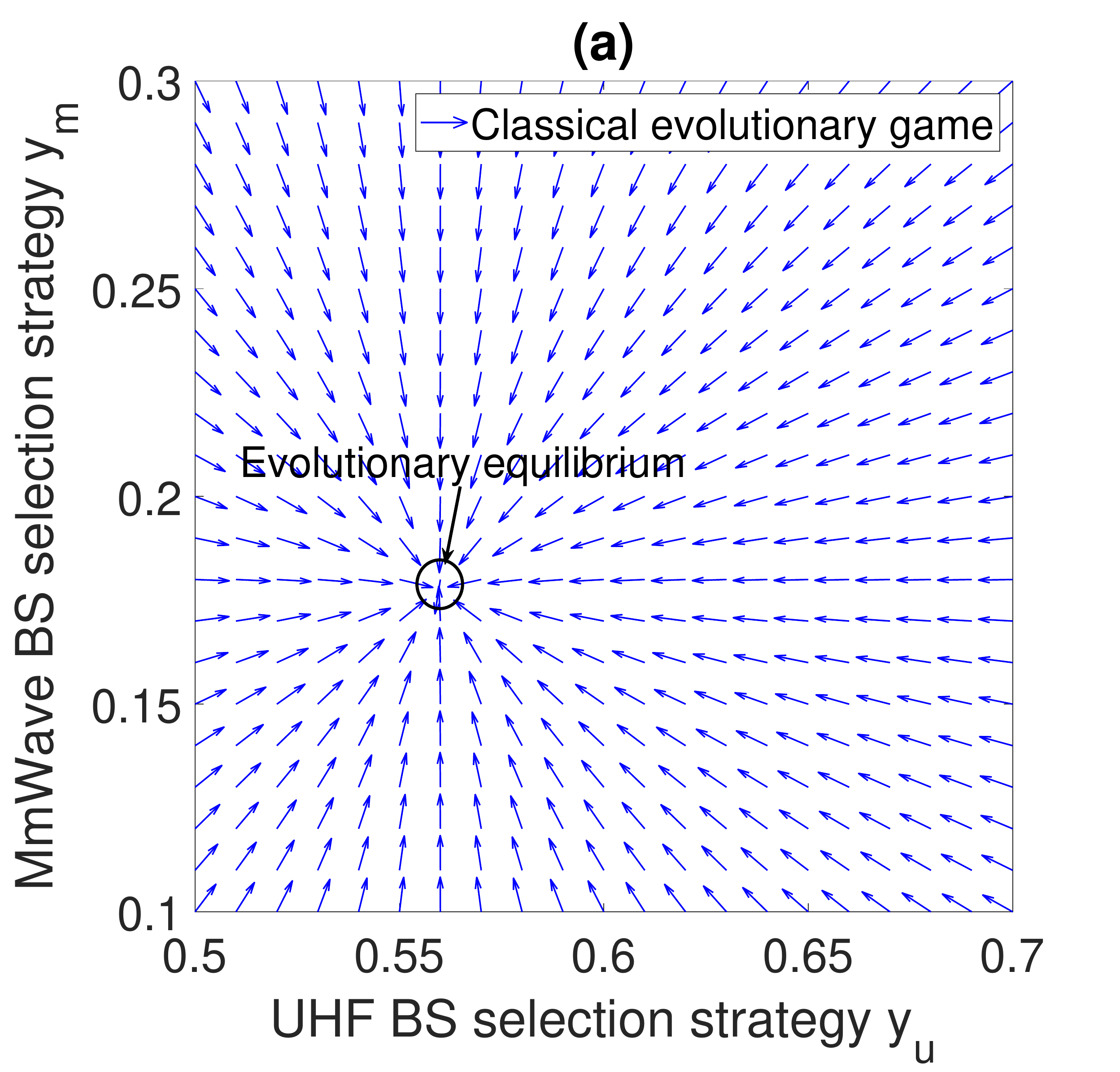}
	\end{minipage}
	\begin{minipage}{5cm}
		\centering
		\includegraphics[width=1.05\textwidth,trim=5 0 10 5,clip]{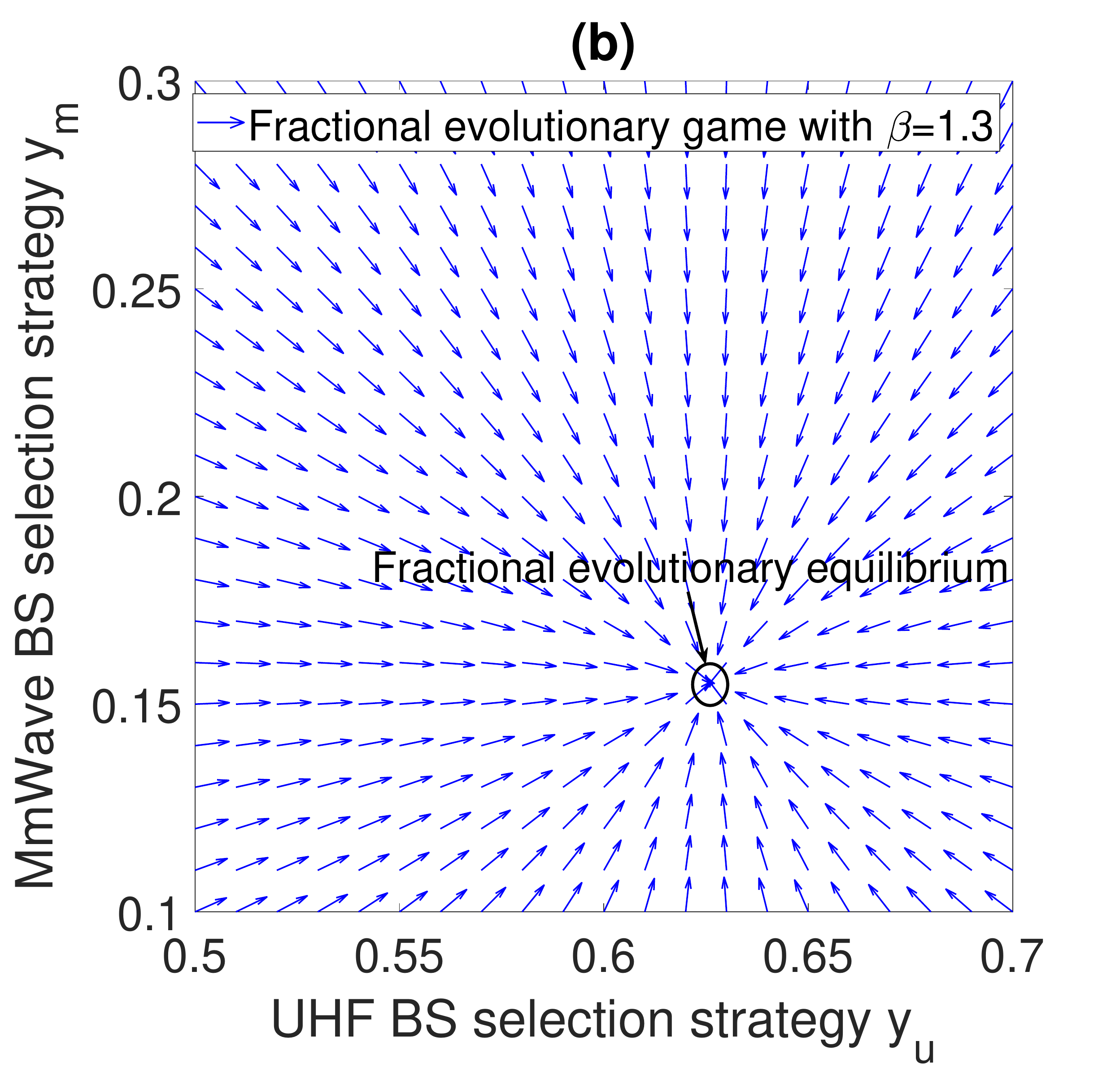}
	\end{minipage}
	\begin{minipage}{5cm}
		\centering
		\includegraphics[width=1.05\textwidth,trim=5 0 10 5,clip]{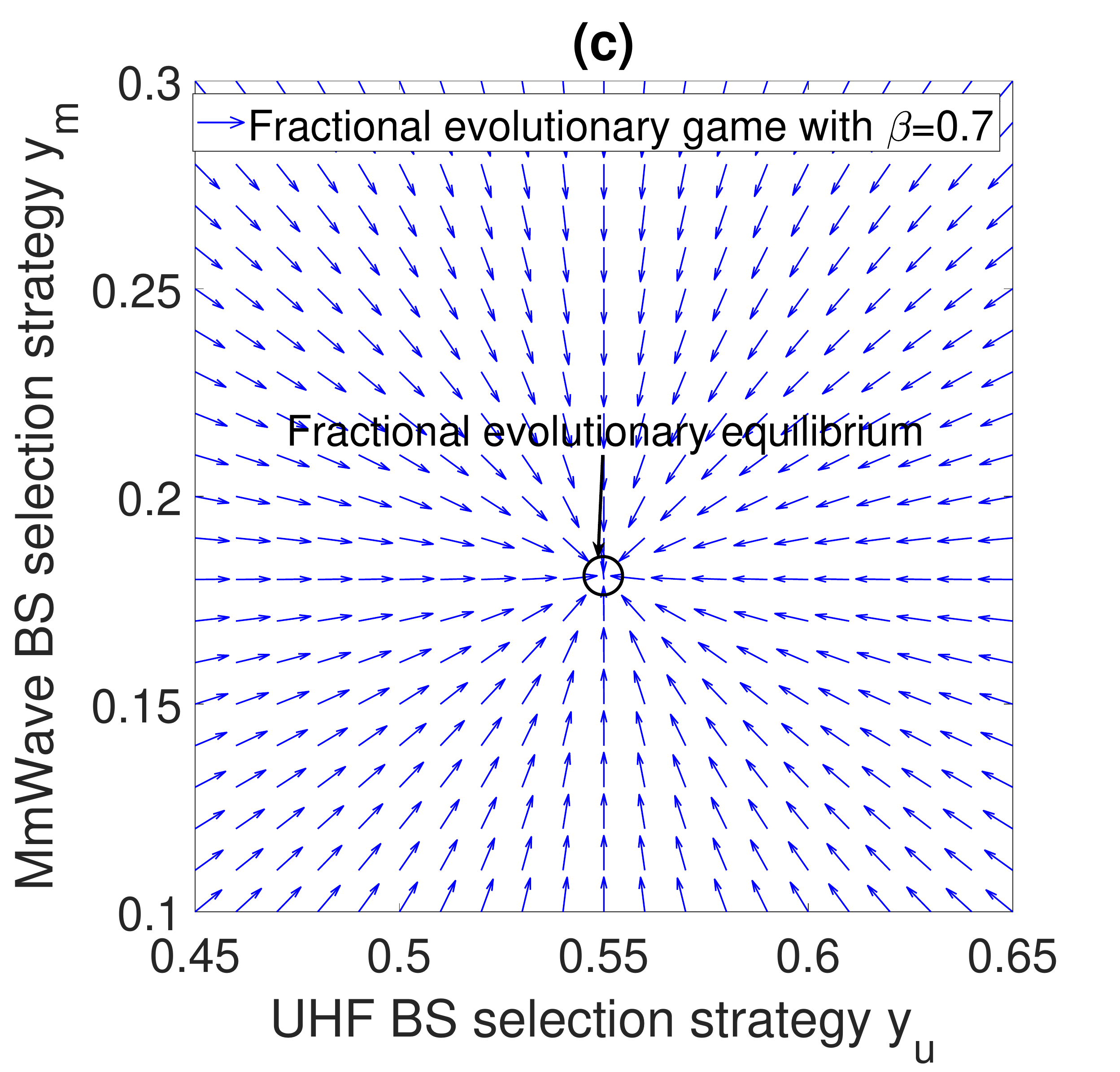}
	\end{minipage}
	\begin{minipage}{5cm}
		\centering
		\includegraphics[width=1.05\textwidth,trim=5 0 10 5,clip]{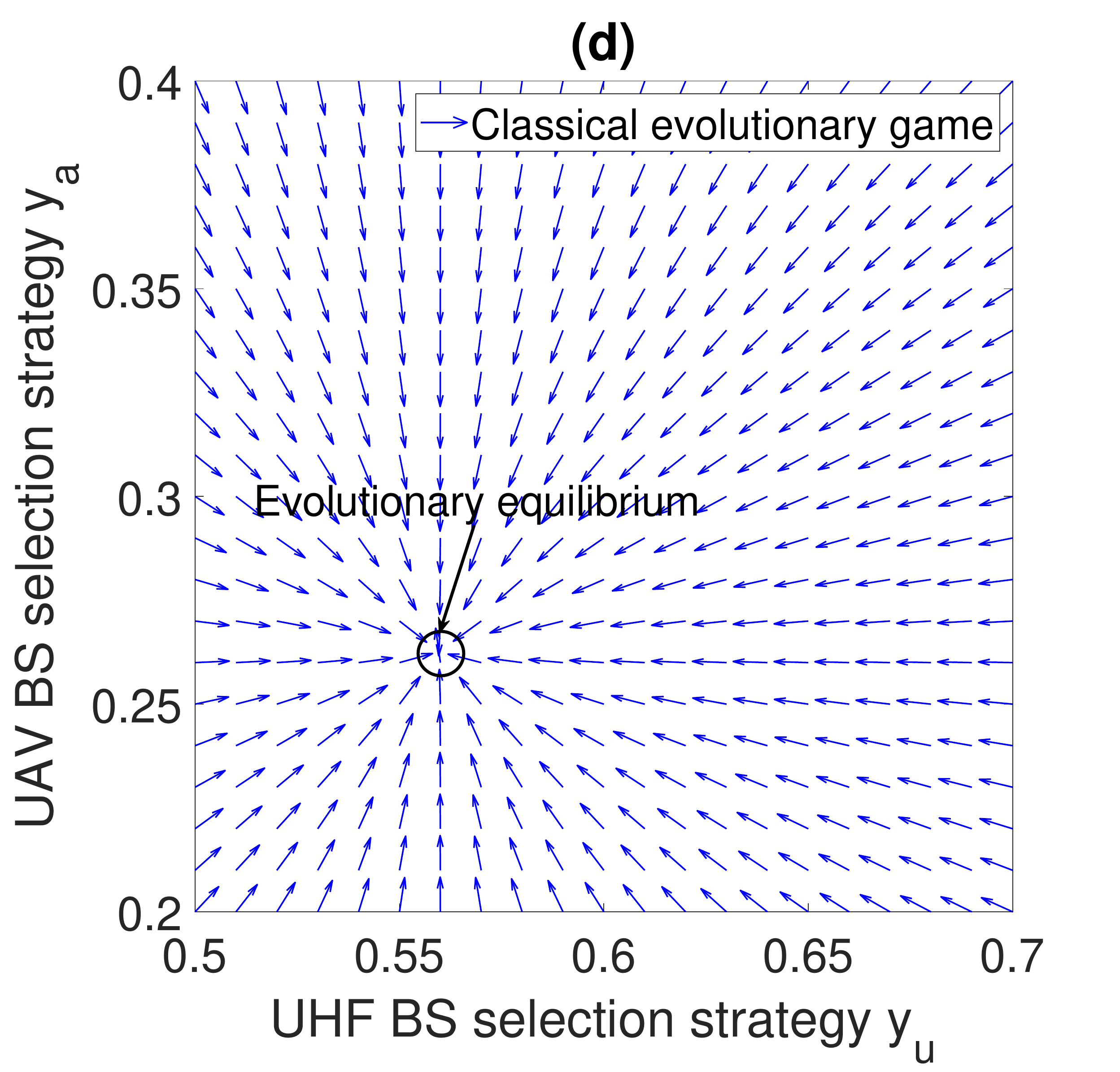}
	\end{minipage}
	\begin{minipage}{5cm}
		\centering
		\includegraphics[width=1.05\textwidth,trim=5 0 10 5,clip]{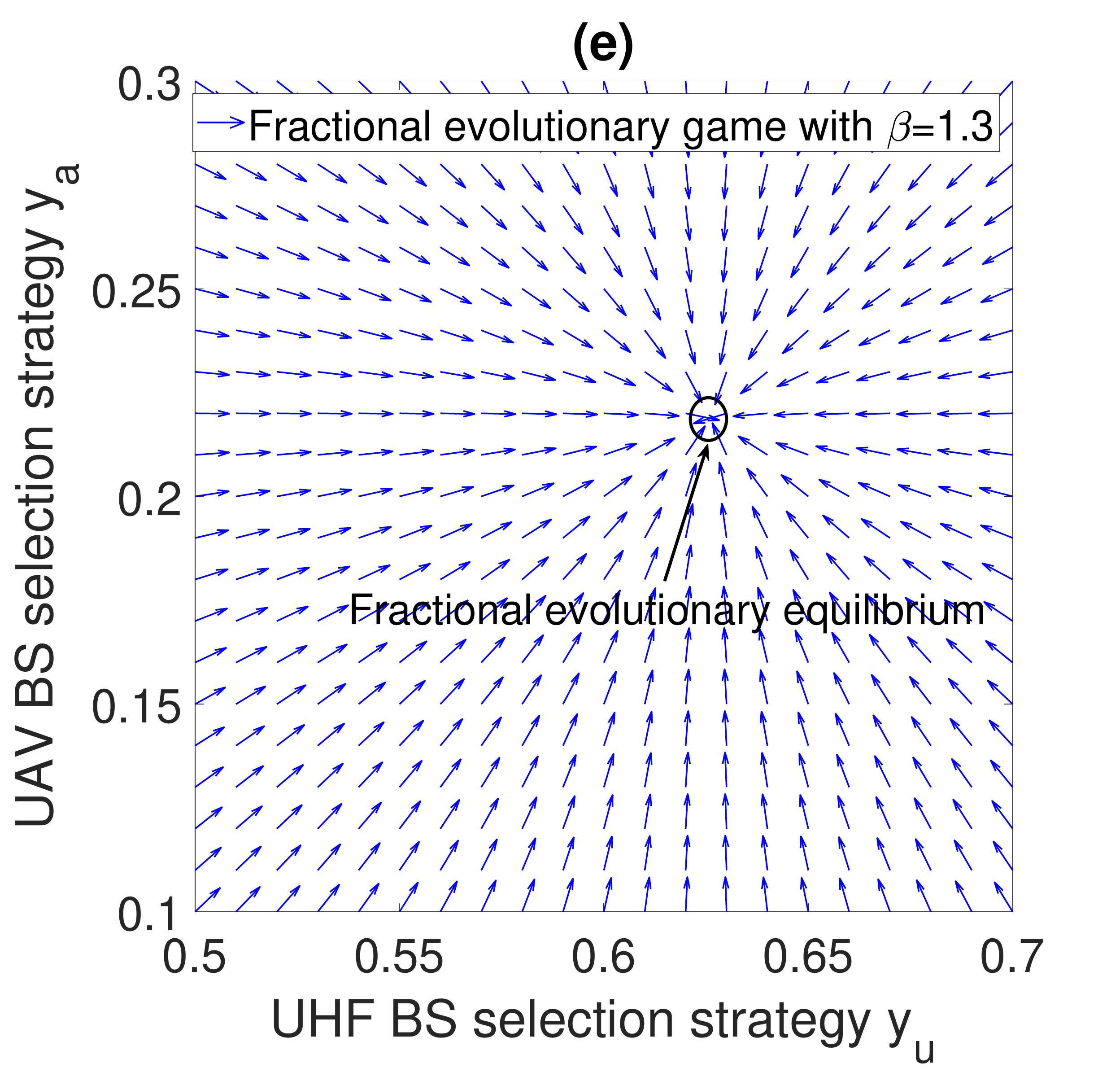}
	\end{minipage}
	\begin{minipage}{5cm}
		\centering
		\includegraphics[width=1.05\textwidth,trim=5 0 10 5,clip]{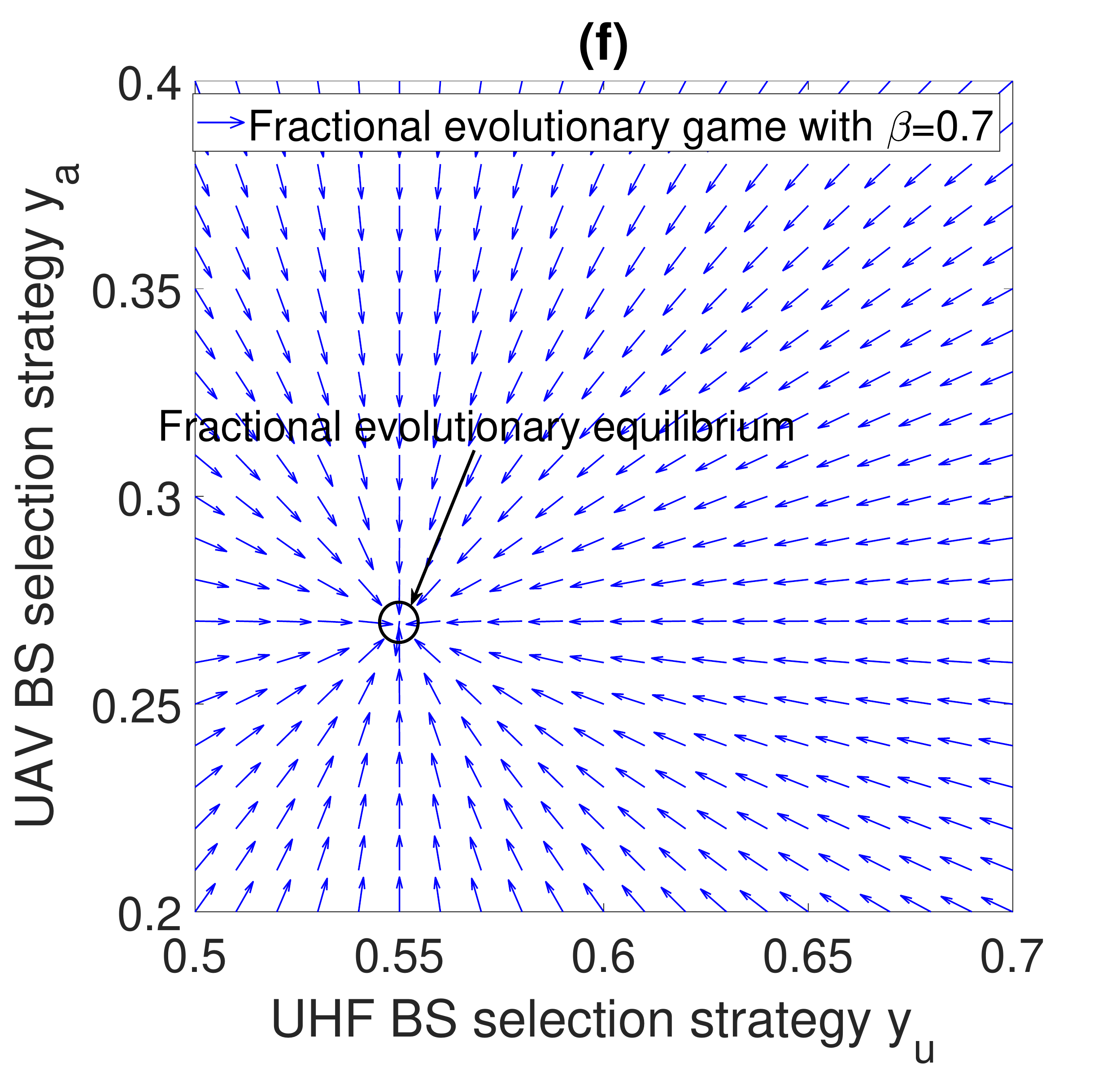}
	\end{minipage}
	\begin{minipage}{5cm}
		\centering
		\includegraphics[width=1.05\textwidth,trim=5 0 10 5,clip]{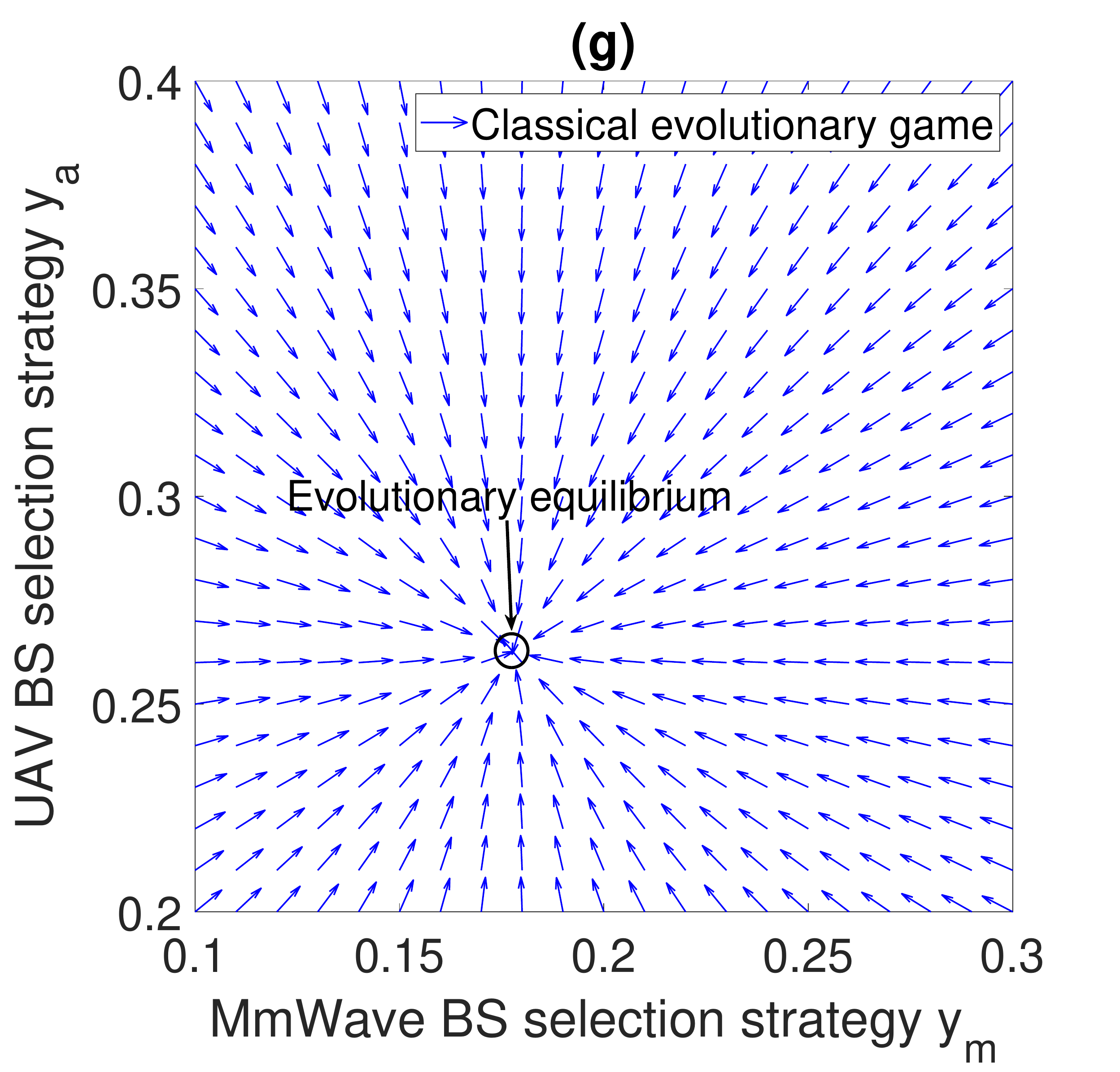}
	\end{minipage}
	\begin{minipage}{5cm}
		\centering
		\includegraphics[width=1.05\textwidth,trim=5 0 10 5,clip]{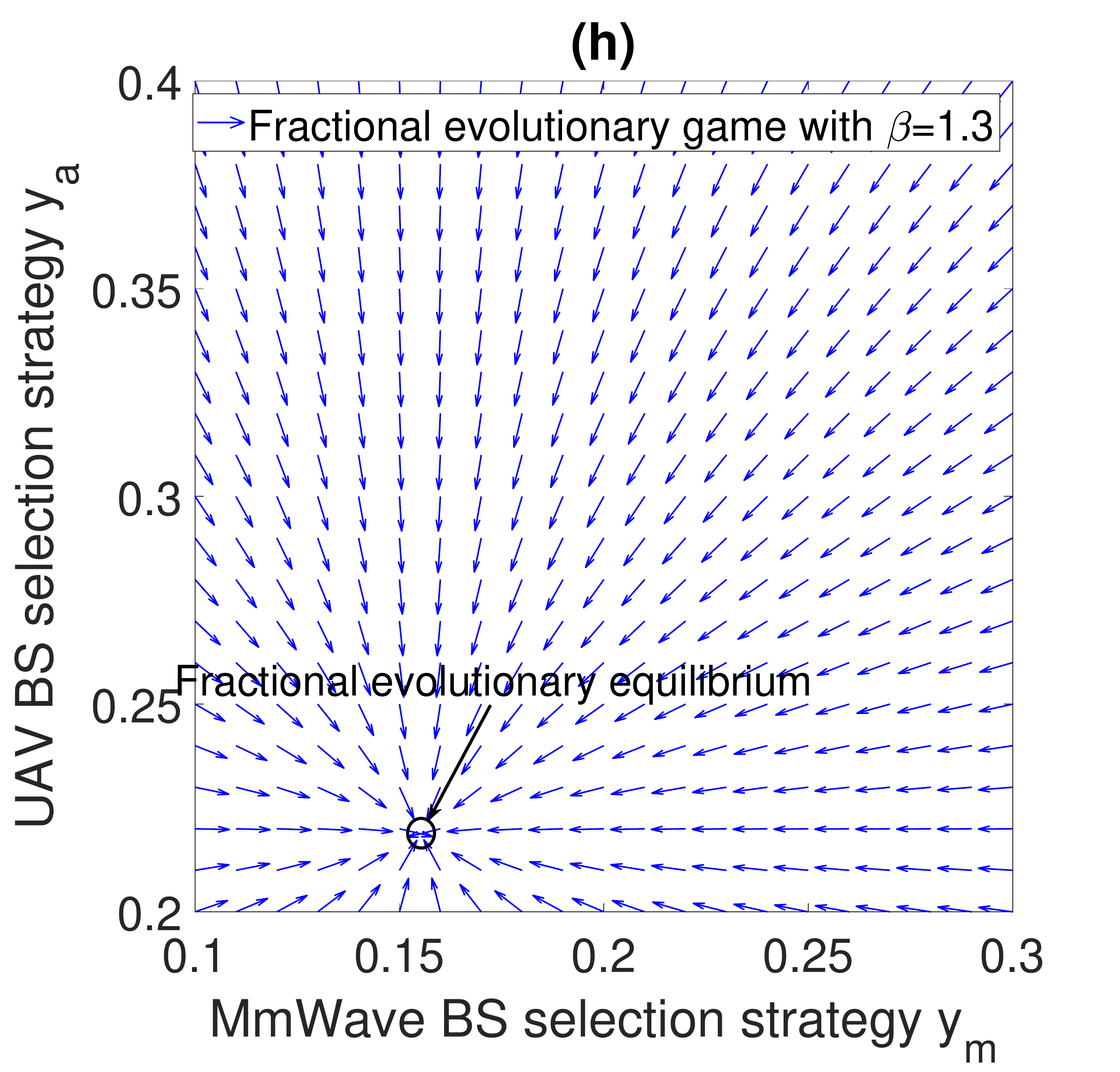}
	\end{minipage}
	\begin{minipage}{5cm}
		\centering
		\includegraphics[width=1.05\textwidth,trim=5 0 10 5,clip]{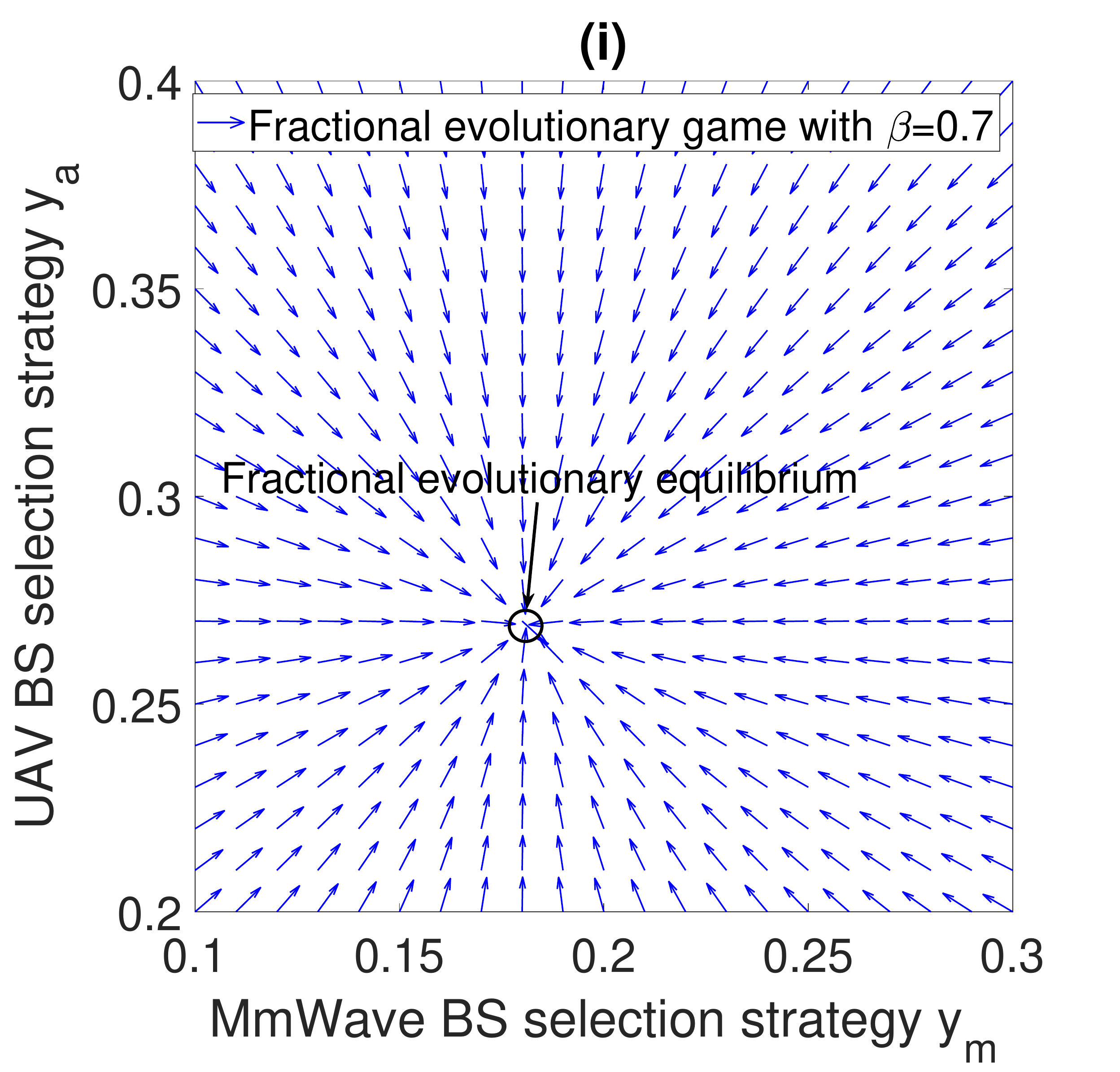}
	\end{minipage}
	\caption{Direction field of the replicator dynamics for verifying the stability of the strategies in the (fractional) evolutionary equilibrium when $t=120$}
	\label{fig:special_case_direction_field}
\end{figure*}

To verify the stability of the strategies in the classical and fractional evolutionary games, we present the direction field of the replicator dynamics. Moreover, to better understand the stability of equilibriums in the classical and fractional evolutionary games, we show the direction field of the replicator dynamics at $t=120$, where the user's strategies have already been stabilized at the equilibrium strategies as shown in Fig.~\ref{fig:special_case_ESS}. As shown in Fig.~\ref{fig:special_case_direction_field}, any unstable strategy will follow the arrow to reach the equilibrium strategy, which is marked by the black circle. This demonstrates the convergence and stability of the strategy and is consistent with Theorem~\ref{th:fractional_equilibrium_stability}. Additionally, the strategies converge to different equilibrium strategies at $t=120$ in the classical and fractional evolutionary games as shown in Fig.~\ref{fig:special_case_ESS}. This can also be verified by using the direction field of the replicator dynamics in Fig.~\ref{fig:special_case_direction_field}, e.g., the unstable strategies follow the arrow to reach different equilibrium strategies in different games as shown in Figs.~\ref{fig:special_case_direction_field}(a), (b), and (c).

%\clearpage

\begin{figure}[!]
	\centering
	\begin{minipage}{6cm}
		\centering
		\includegraphics[width=1.1\textwidth,trim=5 5 15 5,clip]{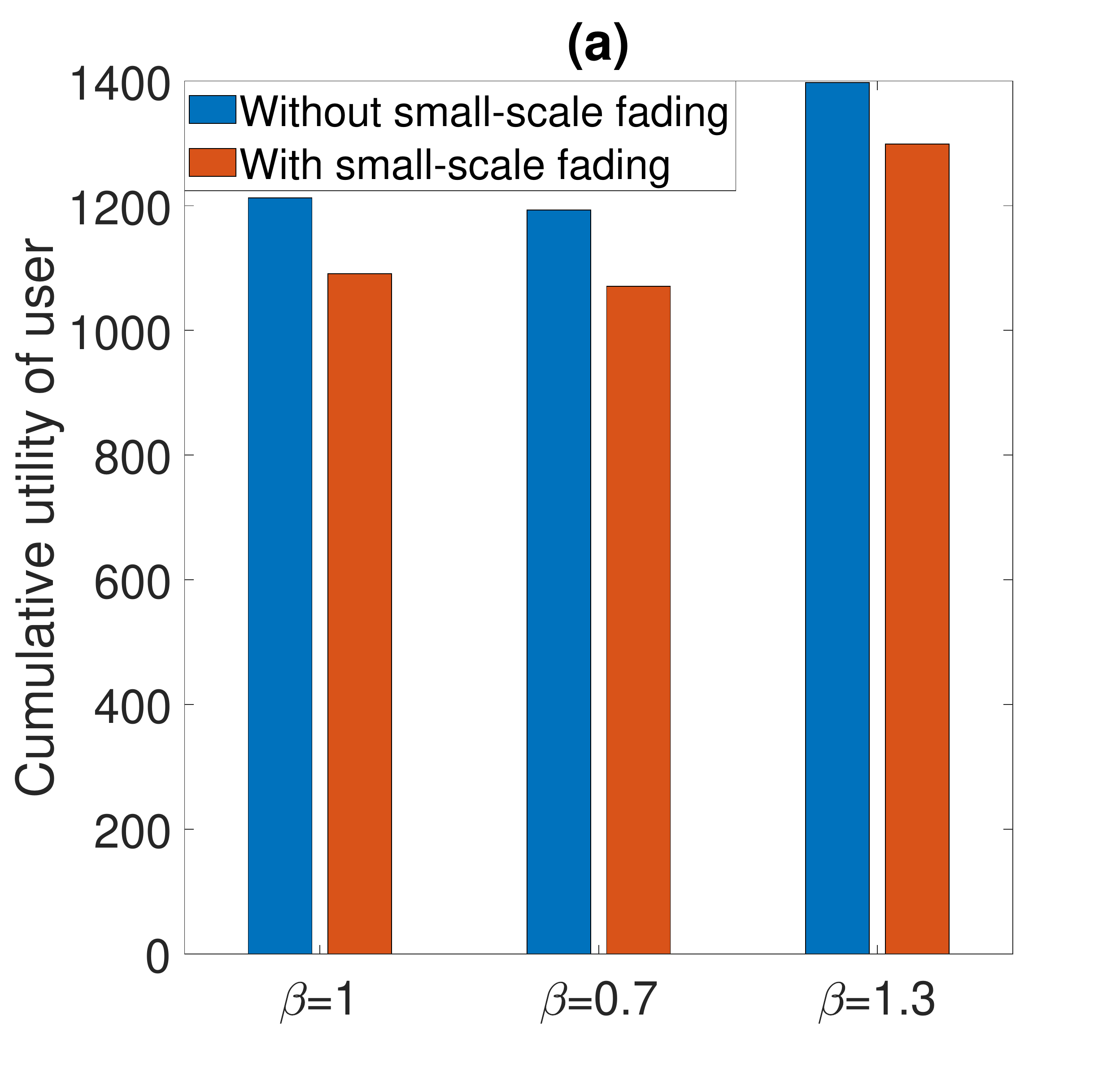}
	\end{minipage}
	\begin{minipage}{6cm}
		\centering
		\includegraphics[width=1.1\textwidth,trim=5 5 15 5,clip]{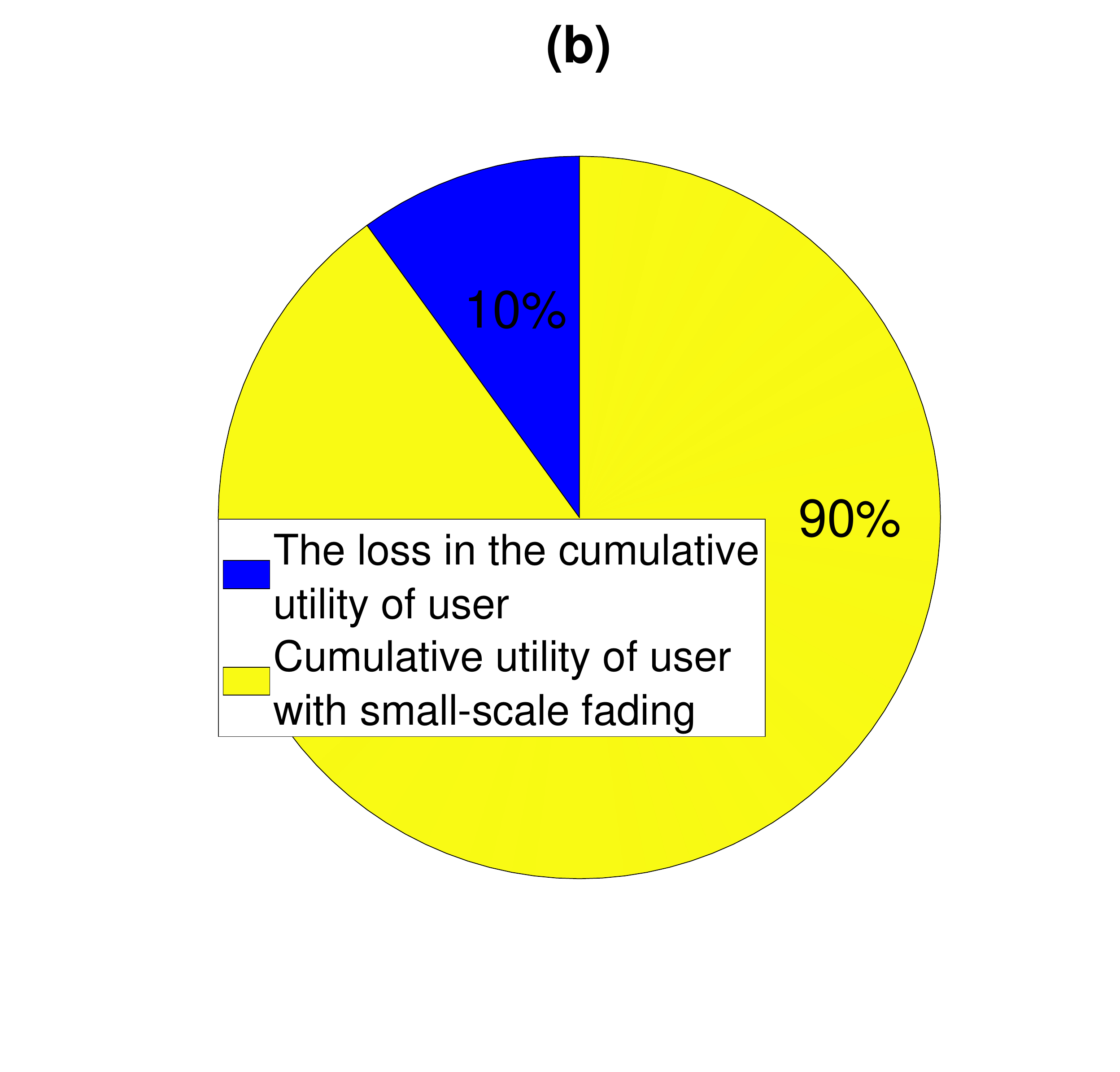}
	\end{minipage}
	\begin{minipage}{6cm}
		\centering
		\includegraphics[width=1.1\textwidth,trim=5 5 15 5,clip]{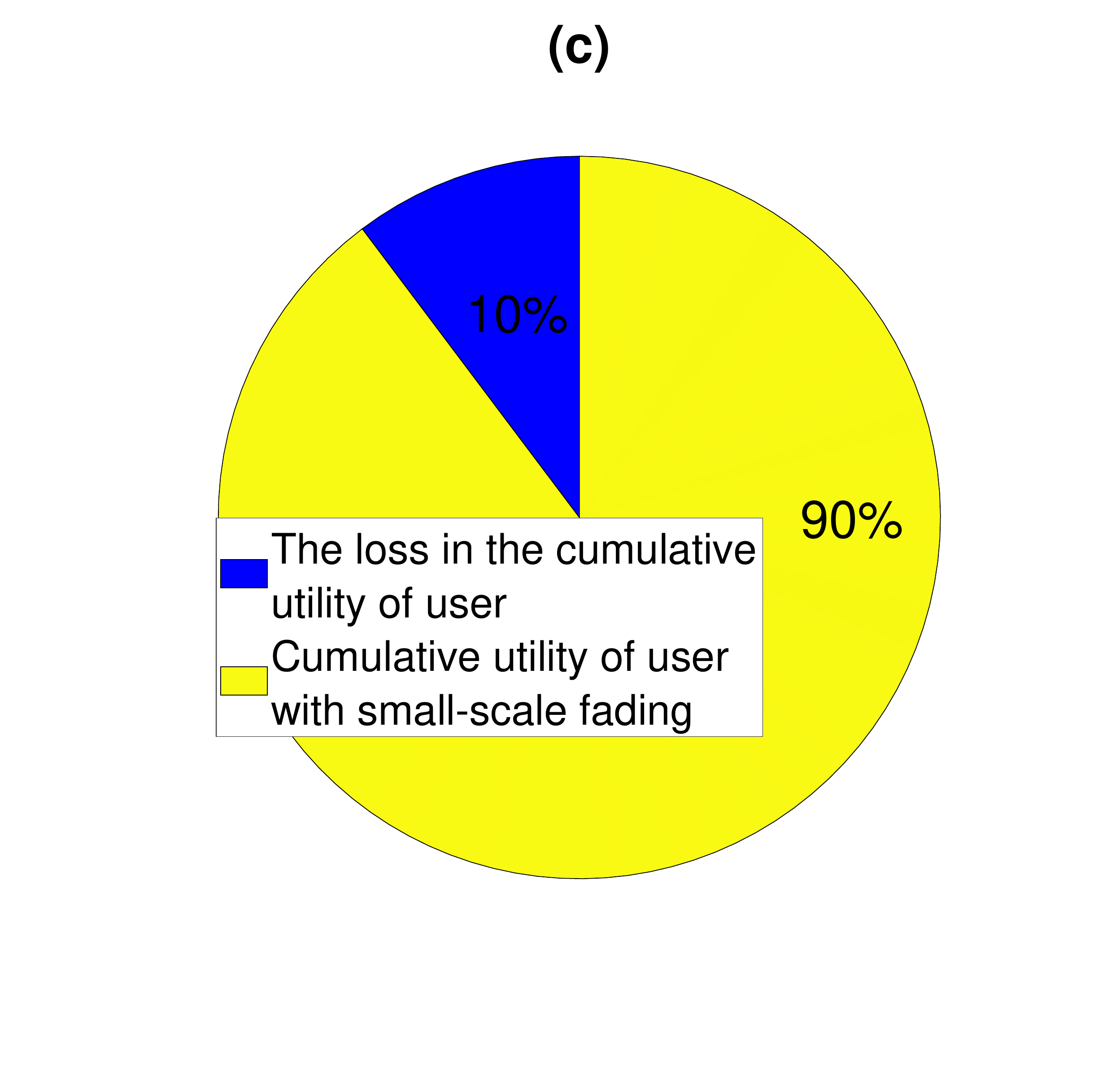}
	\end{minipage}
	\begin{minipage}{6cm}
		\centering
		\includegraphics[width=1.1\textwidth,trim=5 5 15 5,clip]{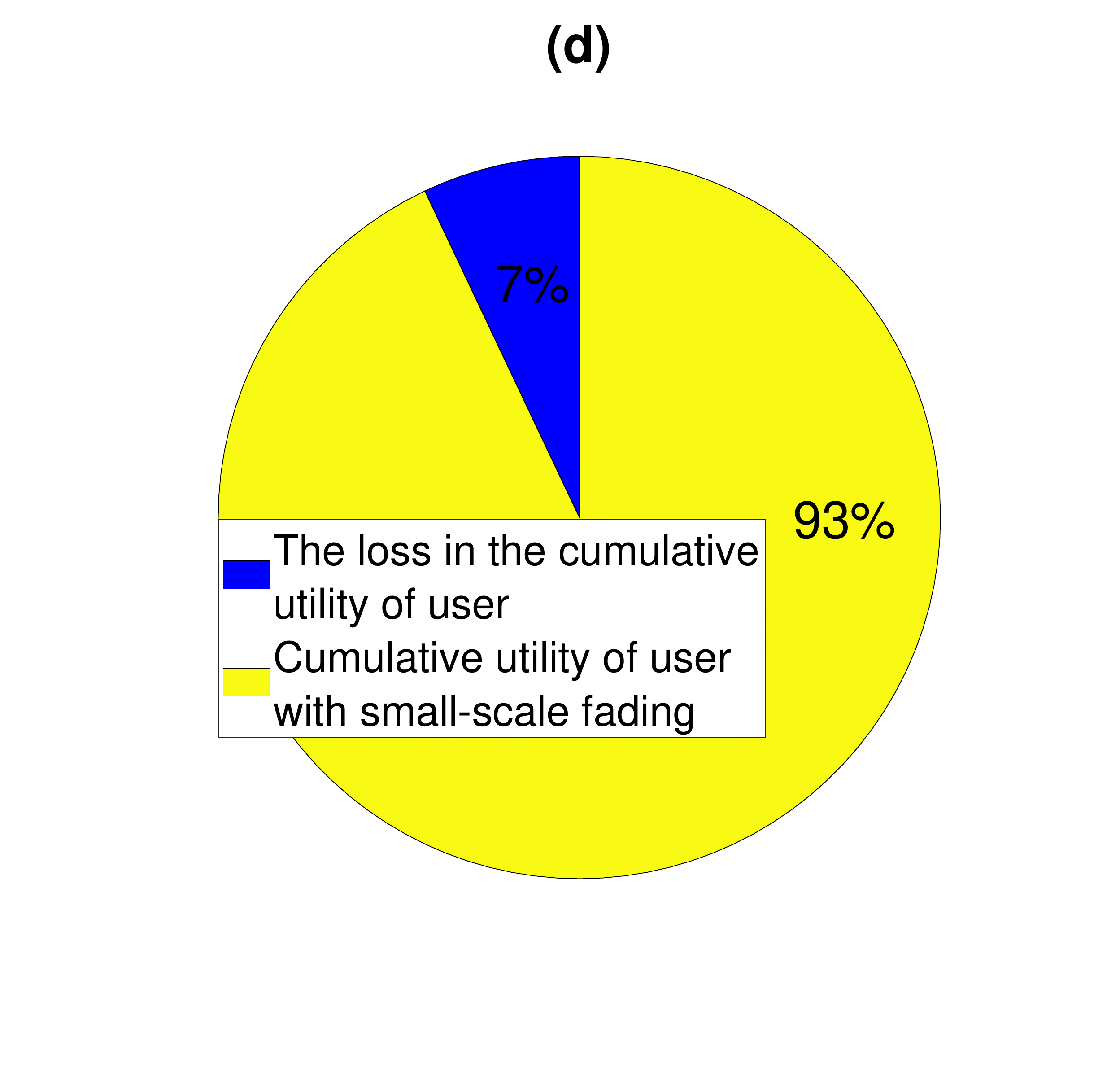}
	\end{minipage}
	\caption{(a) The user's cumulative utility and the loss in the user's cumulative utility in the fractional/classical evolutionary games with $\beta = 1$ (b), $\beta = 0.7$ (c), and $\beta = 1.3$ (d) due to the small-scale fading}
	\label{fig:summary_cumulative_utility}
\end{figure}

Next, we evaluate the joint impact of the memory and the small-scale fading on the user's cumulative utility. We assume the channel gain of the user will fluctuate every time interval of $0.01$. Note here that the results of the user's cumulative utility with the small-fading are the mean of the results we obtained from repeating the fractional/classical evolutionary game for $100$ times. As shown in Fig.~\ref{fig:summary_cumulative_utility}(a), the user's cumulative utilities in all the games get damaged due to the randomicity in the user's channel gain caused by the small-scale fading. To explicitly explore the damage in the user's cumulative utility, we summarize the results of Fig.~\ref{fig:summary_cumulative_utility}(a) in Figs.~\ref{fig:summary_cumulative_utility}(b), (c), and~(d). Here, the results show that the positive memory effect, i.e., $\beta\in\left(1,2\right)$, can help the user to reduce the negative effect caused by the small-scale fading on its cumulative utility compared with the negative memory effect, i.e., $\beta\in\left(0,1\right)$. The reason is that the fast strategy adaptation rate of the users due to the positive memory effect can help the user to timely make decision on its network selection strategy, which in return reduces the negative effect caused by the small-scale fading on its cumulative utility.

%\clearpage

\subsection{Numerical Results for the System Model of the Heterogeneous Network with Heterogeneous Users}
\label{subsec:performance_general}

\begin{table*}[!]
	\centering
	\caption{Location of the BSs and users}
	\begin{tabular}{|c|c|c||c|c|c||}
		\hline
		\hline
		{\bf{Set}} & {\bf{No.}}&{\bf{Coordinate ($Km$)}} &{\bf{Set}} & {\bf{No.}}&{\bf{Coordinate ($Km$)}} \\
		\hline
		\multirow{2}{*}{UHF BS ${\cal{U}}$} & $1$ & $\left[1\ 0\ 0\right]$ &\multirow{4}{*}{\makecell[c]{UAV-enabled \\mmWave BS ${\cal{A}}$}} & $1$ & $\left[0\ 0.1\ 0.02\right]$\\
		\cline{2-3} \cline{5-6}
		& $2$ & $\left[1\ 7\ 0\right]$ & & $2$ & $\left[0\ -0.1\ 0.02\right]$\\
		\cline{1-3} \cline{5-6}
		\multirow{2}{*}{mmWave BS ${\cal{M}}$} &  $1$ & $\left[0\ 0\ 0\right]$ & & $3$ & $\left[0\ 7.1\ 0.02\right]$\\
		\cline{2-3} \cline{5-6}
		& $2$ & $\left[0\ 7\ 0\right]$ & & $4$ & $\left[0\ 6.9\ 0.02\right]$\\
		\hline
		\multirow{8}{*}{User ${\cal{N}}$} & $1$ & $\left[0\ 0.12\ 0\right]$&\multirow{8}{*}{User ${\cal{N}}$} & $9$ &$\left[0\ 7.12\ 0\right]$\\
		\cline{2-3} \cline{5-6}
		& $2$ & $\left[0\ 0.08\ 0\right]$ & & $10$ & $\left[0\ 7.08\ 0\right]$\\
		\cline{2-3} \cline{5-6}
		& $3$ & $\left[0.02\ 0.1\ 0\right]$ & & $11$ & $\left[0.02\ 7.1\ 0\right]$ \\
		\cline{2-3} \cline{5-6}
		& $4$ & $\left[-0.02\ 0.1\ 0\right]$ & & $12$ & $\left[-0.02\ 7.1\ 0\right]$ \\
		\cline{2-3} \cline{5-6}
		& $5$ & $\left[0\ -0.12\ 0\right]$ & & $13$ & $\left[0\ 6.92\ 0\right]$ \\
		\cline{2-3} \cline{5-6}
		& $6$ & $\left[0\ -0.08\ 0\right]$ & & $14$ & $\left[0\ 6.88\ 0\right]$\\
		\cline{2-3} \cline{5-6}
		& $7$ & $\left[0.02\ -0.1\ 0\right]$ & & $15$ & $\left[0.02\ 6.9\ 0\right]$\\
		\cline{2-3} \cline{5-6}
		& $8$ & $\left[-0.02\ -0.1\ 0\right]$ & & $16$ & $\left[-0.02\ 6.9\ 0\right]$\\
		\hline
	\end{tabular}
	\label{tab:BS_user_location_general}
\end{table*}

\begin{table}[!]
	\centering
	\caption{Parameters setting for the heterogeneous network with heterogeneous users}
	\begin{tabular}{|c|c||c|c||}
		\hline
		\hline
		{\bf{Parameters}}& {\bf{Value}} & {\bf{Parameters}}& {\bf{Value}} \\
		\hline
		$\lambda_{u,i}^\phi$& $\frac{1}{3}\times10^{-7}$ & $\lambda_{a,i}^\phi$& $10^{-9}$\\
		\hline
		$\lambda_{m,i}^\phi$& $10^{-9}$ & $\delta$ & $2$\\
		\hline
	\end{tabular}
	\label{tab:parameter_general}
\end{table}

We consider the HetNet with $2$ UHF BSs, $2$ mmWave BSs, and $4$ UAV-enabled mmWave BSs for simplicity, the locations of which are indicated in Table~\ref{tab:BS_user_location_general}. The locations of the users are given in Table~\ref{tab:BS_user_location_general}. The parameter setting for the system model is shown in Table~\ref{tab:parameter_general}. We evaluate the performance of the HetNet with the small-scale fading. 

%\clearpage

\label{subsubsec:performance_general_with_fading}

\begin{figure}[!]
	\centering
	\begin{minipage}{4cm}
		\centering
		\includegraphics[width=1.05\textwidth,trim=5 5 15 5,clip]{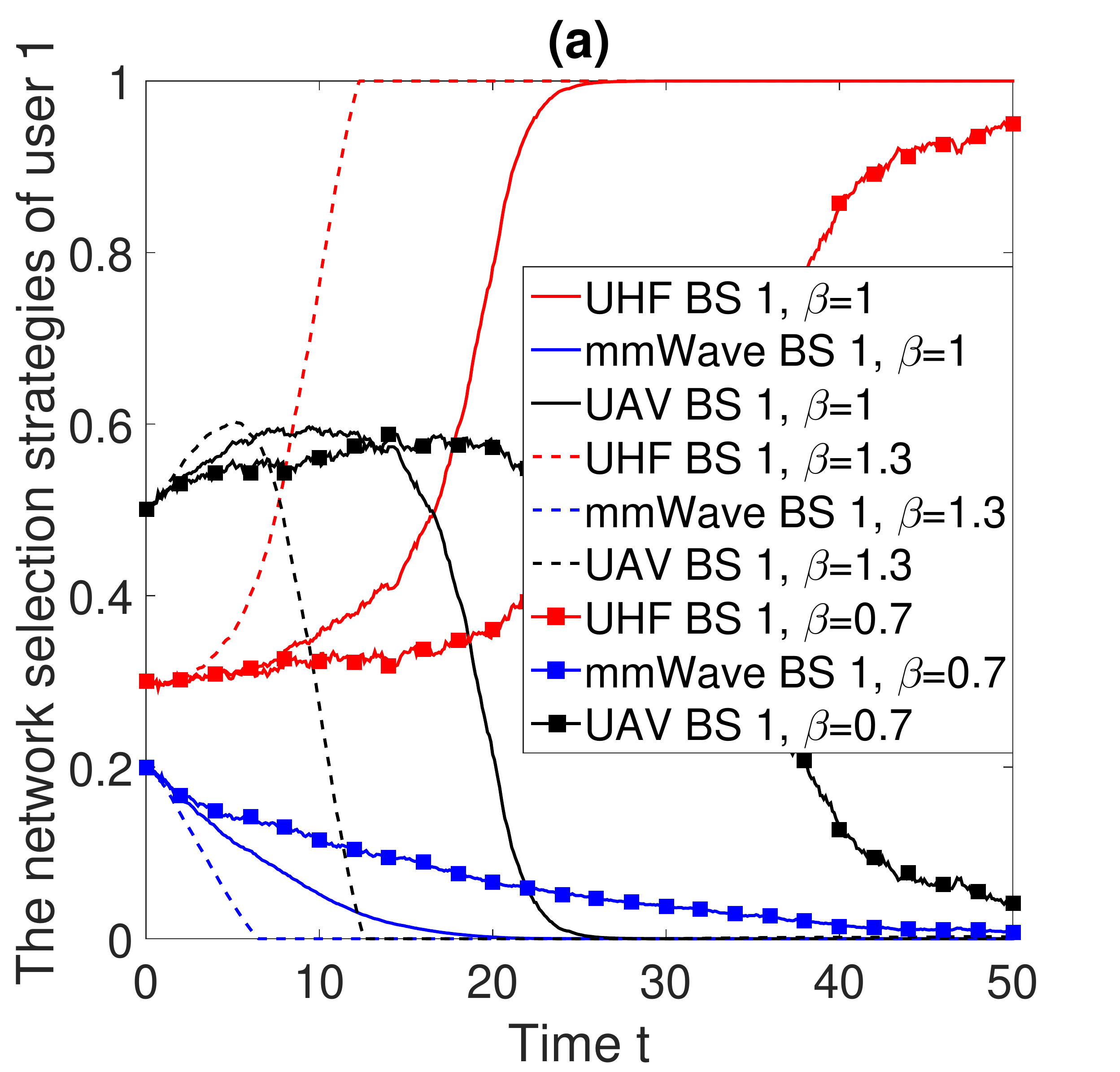}
	\end{minipage}
	\begin{minipage}{4cm}
		\centering
		\includegraphics[width=1.05\textwidth,trim=5 5 15 5,clip]{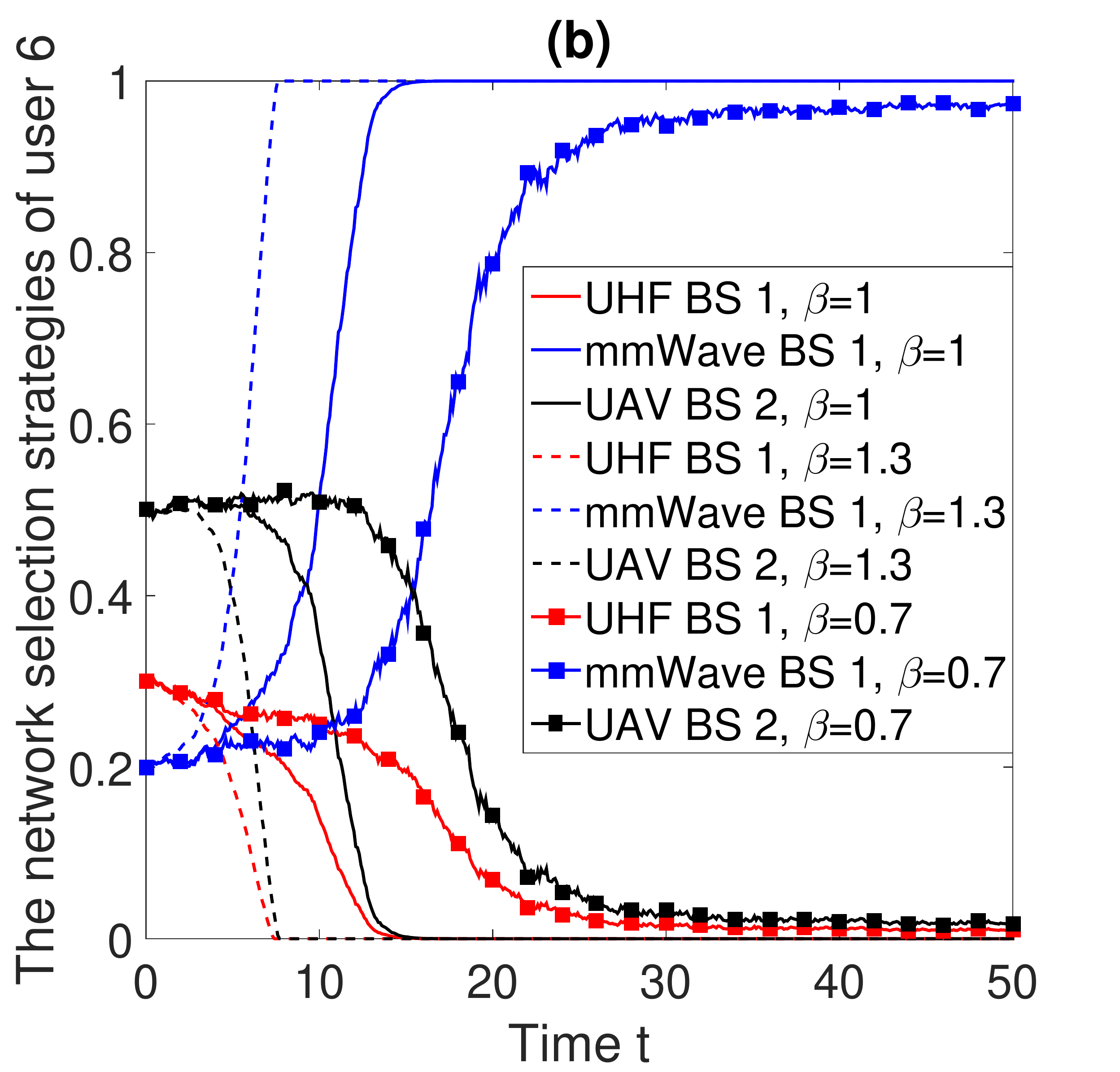}
	\end{minipage}
	\begin{minipage}{4cm}
		\centering
		\includegraphics[width=1.05\textwidth,trim=5 5 15 5,clip]{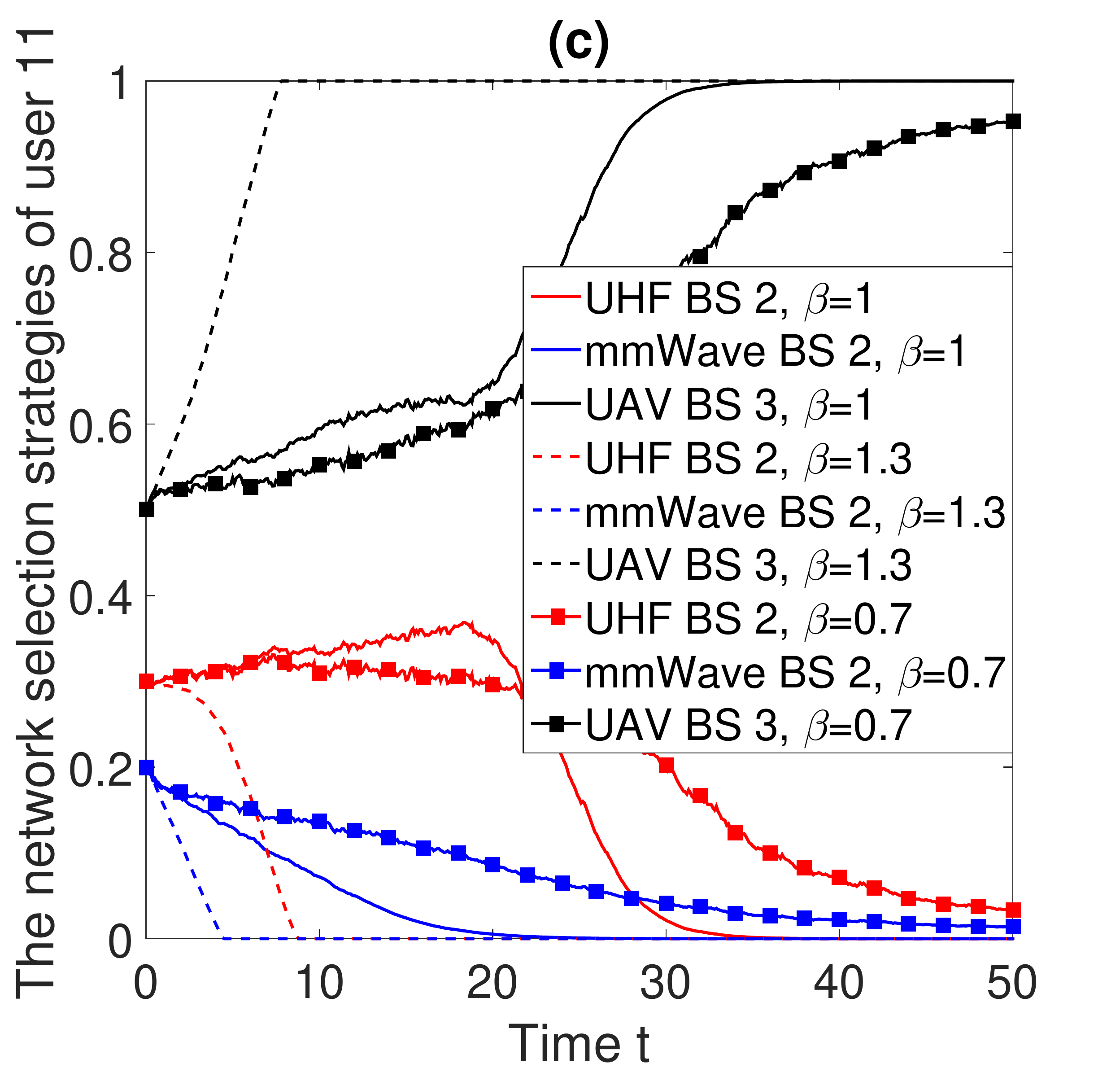}
	\end{minipage}
	\begin{minipage}{4cm}
		\centering
		\includegraphics[width=1.05\textwidth,trim=5 5 15 5,clip]{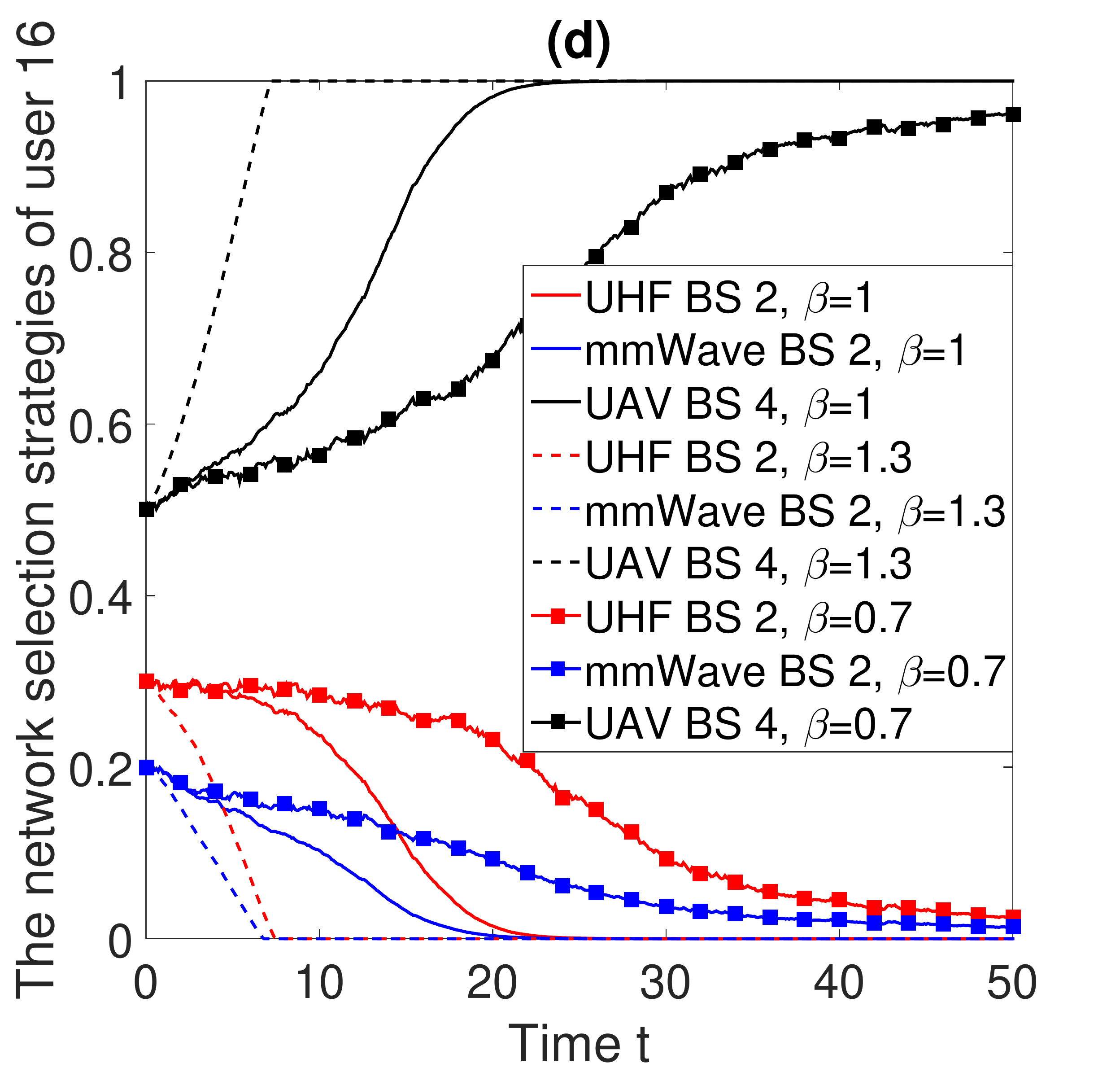}
	\end{minipage}
	\caption{The users' network selection strategies with small-scale fading}
	\label{fig:evolutionary_game_general_uncertain}
\end{figure}

From Fig.~\ref{fig:evolutionary_game_general_uncertain}, we observe that the convergence rate of the replicator dynamics in the fractional evolutionary game with $\beta=1.3$ is still the fastest compared with that in the other two games, i.e., classical evolutionary game and fractional evolutionary game with $\beta=0.7$, under the condition of the small-scale fading. Moreover, the users in the fractional evolutionary game with $\beta=1.3$ are not significantly affected by the small-scale fading. For example, as shown in Fig.~\ref{fig:evolutionary_game_general_uncertain}(c), the user $11$'s UAV-enabled mmWave BS $3$ selection strategy fluctuates in both the classical evolutionary game and fractional evolutionary game with $\beta=0.7$ while that in the fractional evolutionary game with $\beta=1.3$ stably evolves and converges to the equilibrium strategy. The reason is that in the fractional evolutionary game with $\beta>1$, the positive memory effect can incentivize the users to accelerate their reaction speed and further the evolution processes. In this case, the users in the fractional evolutionary game with positive memory effect will experience only a short-term network switch and moreover the negative effect caused by the small-scale fading has little impact on the strategy evolution of the users compared with the users in the other two games.  

%\clearpage

In summary, we have evaluated the performance of the HetNet and the users' utilities under different memory effects. In the HetNet with homogeneous users, the fractional evolutionary game with negative memory effect induces a low utility for the users and slow strategy adaptation for the replicator dynamics compared with the fractional evolutionary game with positive memory effect. Moreover, in the HetNet with homogeneous users, the positive memory effect can help the user to reduce the negative effect caused by the small-scale fading on its cumulative utility. Furthermore, in the HetNet with heterogeneous users, the evolution processes of the users in the fractional evolutionary game with positive memory effect are not significantly affected by the small-scale fading compared with that with negative memory effect.

%\clearpage

\section{Conclusion}
\label{sec:conclusion}

We have presented dynamic game framework to analyze the strategies of the users in the HetNets. The interaction among the users has been first modeled as an evolutionary game, where the dynamic network selection strategies of the users are captured by the replicator dynamics. Then, we cast the classical evolutionary game as a fractional evolutionary game by incorporating the concept of the power-law memory, where the decision-making of the users is affected by their memory. We analytically validated the existence, uniqueness, and stability of the fractional evolutionary equilibrium. We also numerically verified the stability of the fractional evolutionary equilibrium. The stable and unique fractional evolutionary equilibrium has been obtained as the solution to the fractional evolutionary game. Moreover, we have presented a series of insightful analytical and numerical results on the equilibrium of the fractional evolutionary games. For future work, we will further include the UAV-enabled mmWave BSs as players.

%\clearpage
\bibliography{bibfile}

\end{document}